%% file: charge-qubit.tex
\documentclass[a4paper,12pt]{article}

\usepackage[T1]{fontenc}
\usepackage[utf8]{inputenc}
\usepackage[italian, english]{babel}

\usepackage{comment}
\usepackage{appendix}

\usepackage{amsmath}
\usepackage{amssymb}
\usepackage{amsthm}
\usepackage{thmtools,thm-restate}
\usepackage{amsfonts}
\usepackage{mathtools}

\usepackage{braket}

\usepackage{chemformula}

\usepackage{graphicx}
\usepackage{circuitikz}
\usepackage{pgfplots}

\usepackage{subfig}

\newcommand{\numberset}{\mathbb}
\newcommand{\N}{\numberset{N}}
\newcommand{\Z}{\numberset{Z}}
\newcommand{\R}{\numberset{R}}
\newcommand{\C}{\numberset{C}}

\renewcommand{\vec}{\boldsymbol}
\renewcommand{\dag}{^{\dagger}}
\renewcommand{\star}{^*}
\renewcommand{\vec}{\boldsymbol}
\renewcommand{\epsilon}{\varepsilon}
\renewcommand{\rho}{\varrho}

\newtheoremstyle{classicdef}{12pt}{12pt}{}{}{\bfseries}{.}{1em}{}
\newtheoremstyle{classicthm}{12pt}{12pt}{\itshape}{}{\bfseries}{.}{1em}{}

\theoremstyle{classicthm}

\newtheorem{lemma}{Lemma}
\newtheorem{proposition}{Proposition}
\newtheorem{corollary}{Corollary}

\theoremstyle{classicdef}
\newtheorem{definition}{Definition}

\newtheorem{remark}{Remark}

\newcommand{\defsym}{\overset{\mathrm{def}}{=}}

\newcommand{\algebra}{\mathcal{A}}

\newcommand{\identity}{\numberset{I}}

\DeclarePairedDelimiter{\abs} {\lvert} {\rvert}

\DeclarePairedDelimiter{\Bigabs} {\Big\lvert} {\Big\rvert}
\DeclarePairedDelimiter{\norm} {\lVert} {\rVert}
\DeclarePairedDelimiter{\bignorm} {\big\lVert} {\big\rVert}
\DeclarePairedDelimiter{\Bignorm} {\Big\lVert} {\Big\rVert}
\DeclarePairedDelimiter{\expvalue} {\langle} {\rangle}
\DeclarePairedDelimiter{\bigexpvalue} {\big\langle} {\big\rangle}
\DeclarePairedDelimiter{\Bigexpvalue} {\Big\langle} {\Big\rangle}

\newcommand{\Omegastate}{\Omega}

\newcommand{\matrixset}{\mathcal{M}}
\newcommand{\representation}{\pi}
\newcommand{\representationprimed}{\pi'}

\DeclareMathOperator*{\mlim}{m-lim}
\DeclareMathOperator{\bigo}{\mathcal{O}}

\DeclareMathOperator{\Realpart}{Re}
\renewcommand{\Re}{\Realpart}

\newcommand{\omeganb}{\Omega^N_\beta}

\newcommand{\braketmeso}[1]{\bra{n_L'n_R'} #1 \ket{n_Ln_R}}
\newcommand{\Braketgns}[1]{\Bra{n_L'n_R'} #1 \Ket{n_Ln_R}_\beta^N}
\newcommand{\braketgns}[1]{\Bra{n_L'n_R'} #1 \Ket{n_Ln_R}_\beta^N}

\begin{document}

\title{\bf A quantum fluctuation description\\ of charge qubits}

\author{F. Benatti$^{1,2}$, F. Carollo$^{3}$, R. Floreanini$^2$, H. Narnhofer$^4$ and F. Valiera$^5$\\
\\
\small ${}^1$Dipartimento di Fisica, Universit\`a di Trieste, Trieste, 34151 Italy\\
\small ${}^2$Istituto Nazionale di Fisica Nucleare, Sezione di Trieste, 34151 Trieste, Italy\\
\small ${}^3$Institut f\"ur Theoretische Physik, Universit\"at T\"ubingen, Auf der Morgenstelle 14,\\ 
\small72076 T\"ubingen, Germany\\
\small ${}^4$Faculty of Physics, University of Wien, A-1091, Vienna, Austria\\
\small ${}^5$I. Institut f\"ur Theoretische Physik, Universit\"at Hamburg, Notkestraße 9-11, 22607 Hamburg
}
\date{\null}

\maketitle
\input{sections/abstract}
\input{sections/introduction}
\input{sections/charge_qubits}
\input{sections/strong_coupling_BCS}

\input{sections/large_N_formalism}

\input{sections/algebra_fluctuations}

\input{sections/dynamics_fluctuations}

\input{sections/discussion}

\appendix
\input{sections/appendix_single_layer}

\input{sections/appendix_charge_qubits}


\end{document}

%% file: sections/abstract.tex
\abstract{}

\noindent
We consider a specific instance of a superconducting circuit, the so-called charge-qubit, 
consisting of a capacitor and a Josephson junction. Starting from the microscopic description of the latter in terms of two tunneling BCS models in the strong-coupling quasi-spin formulation, we derive the Hamiltonian governing the quantum behavior of the circuit in the limit of a large number $N$ of quasi-spins. Our approach relies on the identification of suitable {\it quantum fluctuations}, {\it i.e.} of collective quasi-spin operators, which account for the presence of fluctuation operators in the superconducting phase  that retain a quantum character in spite of the large-$N$ limit. We show indeed that these collective quantum fluctuations generate the  Heisenberg algebra
on the circle and that their dynamics reproduces the one of the quantized charge-qubit, without the need of a phenomenological ``third quantization'' of a semiclassically inspired model.
As a byproduct of our derivation, we explicitly obtain the temperature dependence of the junction critical Josephson current in the strong coupling regime, a result which is not directly  accessible using standard approximation techniques.

%% file: sections/introduction.tex
\section{Introduction}
\label{intro}
The investigation of emergent properties of quantum many-body systems often requires a description of the system in terms of a suitable 
set of physical observables which account for the behaviour of these systems at large scales. They are usually
observables of collective type, involving all the microscopic degrees of freedom of the system.

Typical collective observables are the so-called mean-field observables 
which account for the average properties of many-body systems. In a system consisting of $N$ microscopic components, mean-field observables are defined as the sum over all components of a same single-component microscopic observable, rescaled by a factor $1/N$.
Their characteristic trait is that, in the large-$N$ limit, these observables form a commutative algebra, thus losing
any quantum character and behaving classically.

Interestingly, there exist other types of collective observables, which still involve sums over all the microscopic components but are rescaled with different powers of $N$. A different rescaling can  allow these observable to retain some quantum 
properties in the large-$N$ limit. For this reason, they are  suitable for modelling quantum phenomena at the \textit{mesoscopic} scale, 
halfway between the quantum microscopic and the classical macroscopic ones.
They have being termed \textit{quantum fluctuations} \cite{Verbeure,Benatti}. 
In many-body systems with sufficiently short-range correlations,  quantum fluctuations are defined, in analogy with a classical central limit theorem, as the 
sums of deviations of microscopic observables from their mean-values, rescaled by $1/\sqrt{N}$. In the large-$N$ limit, these operators generate the Heisenberg algebra of position and momentum operators. 
Different scalings with $N$ of collective system operators are sometimes necessary --- i.e. whenever the state of the system contains long-range correlations --- and give rise to different
non-commutative algebras as $N$ becomes large. All these large-$N$ emerging algebras are referred to as \textit{fluctuation algebras}; they turn out to provide anatural setup for a proper investigation of the mesoscopic properties of many-body systems.
Specifically, the theory of quantum fluctuations has been often employed to describe the \textit{mesoscopic dynamics} of a many-body system, \textit{i.e.}
the quantum dynamics of the fluctuation algebra 
emerging as \hbox{large-$N$} limit of the quantum microscopic many-body dynamics.

An attractive instance of quantum properties emerging in mesoscopic sized many-body systems is provided by the
so-called \textit{superconducting quantum circuits} \cite{Vion}-\cite{Stancil}, 
a quite developed technology for the physical implementation of prototype quantum computers. 
These electric circuits are based on Josephson junctions \cite{Josephson}-\cite{Tafuri},  
namely on two superconducting electrodes, separated by a thin insulating barrier, 
able to sustain tunnelling currents depending non-linearly on the phase difference between the complex order-parameters 
associated with the two electrodes. Remarkably, at low temperatures, superconducting circuits exhibit a quantum behaviour 
that can be effectively described by a phenomenological quantization, a sort of ``third quantization'', 
of their constitutive, classical equations of motion. 
In particular, the excess number of Cooper pairs on the junction and the phase difference between the superconducting condensates 
can be interpreted as conjugated canonical quantum observables,
acting like momentum and angle variables for a particle on a circle inside an anharmonic potential.

Focussing on the specific case of the so-called charge qubit superconducting circuits consisting of a capacitor and a Josephson junction, we shall show that their phenomenological quantum Hamiltonian emerges as the generator of the mesoscopic dynamics of specific quantum fluctuations dictated by the structure 
of the canonical BCS model, without the need of the above-mentioned
``third quantization'' of a semiclassical, phenomenologically inspired model. In this case, quantum fluctuations generate the Heisenberg algebra on the circle,
with collective observables rescaled by $1/N$ playing the role of  
angle-like coordinate operators, while the corresponding conjugate momenta-like operators turn out to be
realized by collective observables without any rescaling.

After reviewing the basic features of a charge qubit superconducting quantum circuit,
we shall analyze the BCS model of superconductivity \cite{Bardeen}-\cite{Abrikosov} 
in its {\it strong-coupling} limit \cite{Thouless1}-\cite{Rodrigues}.
We shall adopt the so-called {\it quasi-spin} formulation \cite{Anderson} and physically motivate the use of the so-called \textit{GNS}-representation of quasi-spin operators; by means of  this technical setting, we shall single out
suitable collective quantum observables, the above-mentioned quantum fluctuations,
thus providing a well-behaved algebraic structure at the mesoscopic scale, as $N$ becomes large \hbox{(Sections 3-5)}. 
This will be the basis for the description of a Josephson junction 
and the derivation of its limiting dynamics, in Section 6.
By inserting the junction in a capacitive circuit, we then directly retrieve the phenomenological charge qubit Hamiltonian 
as the generator of mesoscopic time-evolution.
As a byproduct, the temperature dependence of the critical Josephson current in the strong coupling regime 
is rigorously obtained, a result not easily derivable using standard techniques \cite{Ambegaokar}.
The Appendices contain technical details and proofs.

%% file: sections/charge_qubits.tex
\section{Superconducting charge qubits}
\label{sec:chargequbit}

According to the BCS theory \cite{Bardeen}-\cite{Abrikosov}, 
the transition in metals from the standard to the superconducting phase 
is due to the presence of an attractive, phonon-mediated interaction among electrons that makes two electrons 
having opposite momentum and spin become correlated and form a so-called \emph{Cooper pair}. 
Below a critical temperature $T_c$, Cooper pairs are created on a macroscopic scale, 
giving rise to the appearance of a temperature dependent order parameter
and of an energy gap, corresponding to the energy necessary to create the first excited state of the system. The existence of an energy gap 
explains the absence of dissipation by resistance at temperatures smaller than those able to overcome the gap.

By pairing two superconducting electrodes separated by a thin insulating barrier,
one can then form a Josephson junction \cite{Josephson}-\cite{Tafuri}.
As a consequence of Cooper pairs tunnelling between the barrier, the phases of the (complex) order parameters of the two superconductors
are no longer independent and their difference leads to the generation of an electrical current, the Josephson current, even in absence of an external potential across the barrier.
Nowadays, Josephson junctions are the basic tools for the construction and manipulation of qubits 
using superconducting quantum circuits \cite{Vion}-\cite{Stancil}, 
{\it i.e.} circuits in which electric current, voltage, charge, 
and flux are promoted to quantum observables.

More precisely, a quantum circuit is a network of metallic, insulating, and semiconducting elements, such as capacitors, inductors, Josephson junctions, which, combined with appropriate voltage sources, control the behaviour of the inside circulating current.
The standard approach to modelling these devices consists in first identifying suitable 
canonically conjugated position and momentum like quantities, constructing with them a classical Hamiltonian 
yielding the circuit equations of motion and finally quantizing it.
The whole procedure thus appears as a sort of \emph{third quantization}: the circuit is built with quantum
devices, as are the Josephson junctions, whose macroscopic behavior within the circuit, though rooted in second quantization, is described by a classical, phenomenological Hamiltonian which is then re-quantized. 

In the following, it will be shown that this phenomenolgical prescription can be rigorously motivated and derived by starting from
the microscopic
BCS theory modelling the superconductors involved in the circuits; indeed, through the choice of suitable collective observables
it will be possible to recover, in the limit of a large number of particles, the fully quantized Hamiltonian governing the circuit dynamics,
without the need of a spurious re-quantization.

Although the procedure is general, for sake of concreteness we shall focus on a particular class of superconducting circuits,
those leading to the so-called \emph{charge qubit}. The latter consists of a large superconducting reservoir connected 
to a superconducting island through a Josephson junction. The island is attached to a voltage $V_\mathrm{g}$
through a gate electrode: upon controlling this external voltage, it is then possible to generate an offset charge in the island.
The presence of the junction is required in order to create a non-linearity in the system energy spectrum, so that the two
lowest energy levels can be used to construct a logical qubit.

The equations of motion governing the dynamics of such a circuit can be derived from a phenomenological Hamiltonian
written in terms of an angle variable $\varphi\in[0,2\pi]$, the phase drop across the Josephson junction, and its conjugate momentum $p_\varphi$, measuring the difference in number of Cooper pairs between reservoir and island: 
\begin{equation}
\label{charge-qubit}
H = \mathcal{E}_\mathrm{C}(p_\varphi-n_\mathrm{g})^2-\mathcal{E}_\mathrm{J}\cos(\varphi)\ ,\qquad 
\mathcal{E}_\mathrm{C}\equiv\frac{(2e)^2}{2C}\ ,\qquad\ n_\mathrm{g}\equiv \frac{C_\mathrm{g}V_\mathrm{g}}{2e},
\end{equation}
where the total circuit capacitance $C=C_\mathrm{J}+C_\mathrm{g}$ is the sum of junction, $C_\mathrm{J}$, and gate, $C_\mathrm{g}$,
contributions, while $\mathcal{E}_\mathrm{C}$ and $\mathcal{E}_\mathrm{J}$ are the charging and Josephson energies.
The gate voltage $V_\mathrm{g}$ controls the induced offset charge $n_\mathrm{g}$, where $2e$ 
is the Cooper pair charge. In the charging regime, $\mathcal{E}_\mathrm{C}\gg\mathcal{E}_\mathrm{J}$
only the two lowest-lying charge states, differing by one Cooper pair, are relevant, thus forming
a qubit.

The phenomenologically inspired Hamiltonian (\ref{charge-qubit}) 
is classical and needs to be quantized; in order to do that, one has
to promote angle and momentum variables to quantum operators $\hat\varphi$ and $\hat p_\varphi$,
acting on the Hilbert space $L^2([0,2\pi])$ of periodic square integrable functions on the circle. 
A naive procedure would lead one to introduce the operators:
\begin{equation}
\label{CCR_circle}
    (\hat p_\varphi\psi)(\varphi)\mapsto-i\partial_\varphi\psi(\varphi)\ ,\qquad (\hat\varphi\, \psi)(\varphi)\mapsto\varphi\psi(\varphi).
\end{equation}
However, these operators suffer from the fact that the domain of self-adjointness of $\hat p_\varphi$ 
consists of periodic functions that are no longer periodic when multiplied by $\hat{\varphi}$. Therefore, one has to resort to Weyl-like unitary operators:
\begin{equation}
\label{weyl-circle}
\left(e^{i\alpha \hat p_\varphi}\psi\right)(\varphi)=\psi(\varphi+\alpha),\quad 
\left(e^{i n\hat\varphi}\psi\right)(\varphi)\mapsto e^{i n\varphi}\psi(\varphi),
\end{equation}
with  $n\in\mathbb{Z}$. Then, for $\alpha,\alpha'\in\R$ and $m,n\in\Z$,
\begin{equation}
    e^{i\alpha \hat p_\varphi}\, e^{ in\hat\varphi} = e^{ in\alpha}\, e^{ in\hat\varphi}\, e^{i\alpha \hat p_\varphi},
\label{weyl1}
\end{equation}
which can be equivalently recast as
\begin{equation}
\label{weyl2}
[\hat p_\varphi,e^{i\hat\varphi}] = e^{i\hat\varphi}\ ,\qquad\ e^{-im\hat\varphi}\,\hat p_\varphi\, e^{im\hat\varphi} = \hat p_\varphi+m\ ,
\quad m\in\Z\ .
\end{equation}
%
In this way, the angular momentum operator $\hat{p}_\varphi$ and the Weyl phase operator form the \emph{Heisenberg algebra on the circle}.
It can be represented on the Hilbert space $L^2\bigl([0,2\pi]\bigr)$ equipped with 
the orthonormal basis $\{|n\rangle\}$, $n\in\Z$, of ``plane waves'',
\begin{equation}
\langle\theta |n\rangle = \frac{1}{\sqrt{2\pi}} e^{i n \theta}\ .
\label{state-n}
\end{equation}
These vectors are obtained by acting on the ``vacuum" state $\langle\theta |0\rangle=1/\sqrt{2\pi}$, $\hat{p}_\varphi\ket{0}=0$,
$$
|n\rangle = e^{i n \hat{\varphi}} |0\rangle\ ,
$$
and, because of~\eqref{weyl2}, are
eigenstates of the momentum $\hat{p}_\varphi$, $\hat{p}_\varphi |n\rangle= n |n\rangle$, satisfying the orthogonality condition:
\begin{equation}
\label{scalarproduct}
\braket{n|m}=\braket{0 | e^{i(m-n)\hat{\varphi}} | 0} = \frac{1}{2\pi}\int_0^{2\pi} \, d\theta e^{i(m-n)\theta} = \delta_{m,n}\ .
\end{equation}

%% file: sections/strong_coupling_BCS.tex
\section{Strong coupling limit of the BCS model}
\label{sec:strongcoupling}

Superconductivity can be described as a microscopic phenomenon originated  by the condensation of bounded pairs of electrons, the
Cooper pairs. The low-energy pairing physics is very well captured by the BCS Hamiltonian \cite{Bardeen}-\cite{Abrikosov}:
\begin{equation}
\label{BCS}
H_{BCS} = \sum_{\vec{k}}\epsilon_{\vec{k}}(c\dag_{\vec{k}\uparrow}c_{\vec{k}\uparrow}
+c\dag_{-\vec{k}\downarrow}c_{-\vec{k}\downarrow})-
\sum_{\vec{k}}\, V_{ \vec{k} \vec{k}'}\,
c\dag_{\vec{k}'\uparrow}c\dag_{-\vec{k}'\downarrow}c_{-\vec{k}\downarrow}c_{\vec{k}\uparrow} \ ,
\end{equation}
where $c_{\vec{k}\sigma}$ are fermionic annihilation operators of electronic states of momentum $\vec{k}$, spin $\sigma$
and energy $\epsilon_{\vec{k}}$, while $V_{ \vec{k} \vec{k}'}$ describes an attractive two-body interaction.

An often adopted approximation to this Hamiltonian assumes the pairing potential $V_{ \vec{k} \vec{k}'}$
to be a positive constant $V\geq 0$ for all indices $\vec{k}$, $\vec{k}'$ in a region around the Fermi energy
determined by the cutoff $\omega_\mathrm{D}$, the Debye energy of the crystal where the electrons move,
and vanishing outside this region. As the interaction term acts only within the cutoff 
region around the Fermi energy $\epsilon_\mathrm{F}$, focusing on the physics relative to this energy region, 
one can make an additional simplification and take all the single-electron energy levels 
$\epsilon_{\vec{k}}$ around the Fermi energy to be constant. This is the essence of the so-called
{\it strong-coupling} approximation \cite{Thouless1}-\cite{Rodrigues}, which is justified when $\omega_\mathrm{D}\ll V$
so that the $\vec{k}$-dependence in $|\epsilon_{\vec{k}} - \epsilon_\mathrm{F}| < \omega_\mathrm{D}$ 
can be safely ignored.

In addition, notice that the Hamiltonian (\ref{BCS}) involves only pair operators of the form
$c_{-\vec{k}\downarrow}c_{\vec{k}\uparrow}$;
they generate an $su(2)$ algebra and therefore they can be represented by quasi-spin matrices \cite{Anderson}:
\begin{equation}
\label{su2}
c_{-\vec{k}\downarrow}c_{\vec{k}\uparrow}=\sigma^{(\vec{k})}_+\ ,\quad 
(c_{-\vec{k}\downarrow}c_{\vec{k}\uparrow})^\dagger=\sigma^{(\vec{k})}_-\ ,\quad
c\dag_{\vec{k}\uparrow}c_{\vec{k}\uparrow} + c\dag_{-\vec{k}\downarrow}c_{-\vec{k}\downarrow}=1-\sigma^{(\vec{k})}_z\ ,
\end{equation}
where $\sigma_i$, $i=x,y,z,$ are Pauli matrices, and $\sigma_\pm=(\sigma_x \pm i \sigma_y)/2$.
By enumerating the finite number $N$ of available energy levels in the cutoff region with an integer index $k$, 
the strong-coupling limit of the BCS Hamiltonian (\ref{BCS}) 
can finally be rewritten in the quasi-spin language as \cite{Thirring2,Thirring3}
\begin{equation}
\label{BCS-spin}
H^N= -\epsilon\sum_{k=1}^N\sigma^{(k)}_z-\frac{2\,T_c}{N}\sum_{k,k'=1}^N\sigma^{(k)}_+\sigma^{(k')}_-,
\end{equation}
up to an irrelevant constant contribution. For later convenience we have set $V=2\,T_c/N$, 
where the constant $T_c$ plays the role of the
critical temperature. Equivalently, in terms of collective spin operators
\begin{equation}
\label{collective-spin}
    S^N_i = \frac{1}{2}\sum_{k=1}^N\sigma_i^{(k)}\ ,\quad i=x,y,z,\qquad
    S^N_\pm = S^N_x\pm iS^N_y\ ,
\end{equation}
one has:
\begin{equation}
\label{BCS-spin-2}
H^N= -2\epsilon\, S^N_z-\frac{2\,T_c}{N}\, S^N_+\, S^N_-\ .
\end{equation}
Although a simplified version of the original
BCS Hamiltonian, the model based on (\ref{BCS-spin}), (\ref{BCS-spin-2}) is able to capture the relevant characteristic features
of superconducting devices, including Josephson junctions \cite{Thirring2,Rodrigues,Unnerstall,Lauwers}.

In particular, the Hamiltonian (\ref{BCS-spin-2}) 
can be treated as in the standard case by considering its mean-field approximation,
\begin{equation}
\widetilde H^N= -2\epsilon\, S^N_z-2\,T_c\, \Big( S^N_+\, \langle\!\langle S^N_-\rangle\!\rangle_\beta^N 
+ \langle\!\langle S^N_+\rangle\!\rangle_\beta^N \, S^N_-\Big)\equiv -\sum_{k=1}^N \sum_{\nu=1}^3 \omega_\nu\cdot \sigma_\nu^{(k)}\ ,
\label{mean-field}
\end{equation}
where $\omega_1=\omega_2=T_c$, $\omega_3=\epsilon$ and $\langle\!\langle \,\cdot\,\rangle\!\rangle_\beta^N$ is the expectation value with respect to the mean-field 
Gibbs state at the inverse temperature $\beta$:
\begin{equation}
\langle\!\langle \cdot\rangle\!\rangle_\beta^N = 
\frac{ {\rm Tr}\big[ e^{-\beta\, \widetilde H^N} \,\cdot\,\big]}{{\rm Tr}\big[ e^{-\beta\, \widetilde H^N}\big]}\ .
\label{mf-average}
\end{equation}
The required consistency condition,
\begin{equation}
\beta_c\,\omega = \tanh (\beta\,  \omega)\ ,\qquad \omega = |\vec\omega|=\sqrt{\epsilon^2+\,4\,T_c^2\,\Delta^2}\ ,
\label{consistency}
\end{equation}
allows for a non-vanishing value of the ``gap'',
\begin{equation}
\frac{\langle\!\langle S^N_+\rangle\!\rangle_\beta^N}{N}\equiv \Delta\, e^{i\phi}\ ,
\label{gap}
\end{equation}
only for temperature $T<T_c=1/\beta_c$, thus recovering the standard result. 

Notice that, as defined in~\eqref{gap}, the gap  is dimensionless and that the relation (\ref{consistency}) fixes its modulus, $\Delta$, but not its phase $\phi$. This is to be expected
as the Hamiltonian (\ref{BCS-spin-2}) is invariant under rotations around the $z$-axis, corresponding to the gauge
invariance of the original BCS Hamiltonian, while the three-dimensional vector $\vec\omega$ in (\ref{mean-field})
points to a given direction.
Remarkably, gauge invariance can be recovered, at least in the large-$N$ limit, by averaging over $\phi$; indeed,
the expectation in the Gibbs state constructed with the Hamiltonian (\ref{BCS-spin-2}),
\begin{equation}
\langle \,\cdot\,\rangle_\beta^N = 
{\rm Tr}\big[\rho_\beta^N\, \cdot \big]\ ,
\qquad \rho_\beta^N=\frac{e^{-\beta\,  H^N}}{{\rm Tr}\big[ e^{-\beta\,  H^N}\big]}\ ,
\label{Gibbs-state}
\end{equation}
can be obtained as \cite{Thirring2}:
\begin{equation}
\lim_{N\to\infty}\langle \,\cdot\,\rangle_\beta^N =
\lim_{N\to\infty}\frac{1}{2\pi}\int_0^{2\pi} \! d\phi \, \langle\!\langle \,\cdot\,\rangle\!\rangle_\beta^N\ ,
\label{Thirring}
\end{equation}
where $\langle\!\langle \,\cdot\,\rangle\!\rangle_\beta^N$ in the r.h.s. is computed with a mean-field Hamiltonian
(\ref{mean-field}) having a gap with phase $\phi$.

As an application, let us compute the large-$N$ limit of the average 
$\langle S^N_+S^N_-\rangle_\beta^N/N^2$, where the rescaling with $N$ is necessary 
in order to obtain a finite result.
Applying (\ref{Thirring}), one gets:
\begin{equation*}
\mathfrak{c}^2 \equiv \lim_{N\to\infty} \frac{1}{N^2}\langle S^N_+S^N_-\rangle_\beta^N=
\lim_{N\to\infty}\frac{1}{2\pi}\int_0^{2\pi} d\phi\, \frac{1}{N^2}\langle\!\langle S^N_+S^N_-\rangle\!\rangle_\beta^N.
\end{equation*}
In order to compute the mean-field expectation in the r.h.s., let us consider the splitting:
\begin{equation*}
\frac{1}{N^2}S^N_+S^N_-
= \frac{1}{N^2}\sum_{k=1}^N\sigma^{(k)}_+\sigma^{(k)}_- 
+ \frac{1}{N^2} \sum_{ \substack{k,k'=1\\ k\ne k'} }^N\sigma^{(k)}_+\sigma^{(k')}_-\ .
\end{equation*}
The first sum contains $N$ terms, each one a projection $\sigma^{(k)}_+\sigma^{(k)}_-=(1+\sigma_z^{(k)})/2$; as a consequence, the factor $1/N^2$ 
makes this term vanish in the large-$N$ limit.
Instead, the second sum contains $N(N-1)$ terms, each one involving two different spin operators, thus making it 
generally non-vanishing. Since in the mean-field theory different sites are statistically independent 
and the Hamiltonian is permutation invariant, the mean-field average of $\sigma_i^{(k)}$ result independent on $k$;
therefore,
\begin{equation*}
\frac{1}{N^2} \sum_{ \substack{k,k'=1\\ k\ne k'} }^N \langle\!\langle\sigma^{(k)}_+\sigma^{(k')}_-\rangle\!\rangle_\beta^N
= \langle\!\langle\sigma_+\rangle\!\rangle_\beta^N\ \langle\!\langle\sigma_-\rangle\!\rangle_\beta^N\ ,
\end{equation*}
which in turn depends only on the gap modulus, $\Delta$, and not on its phase; thus, finally:
\begin{equation}
\mathfrak{c} = \Delta\ .
\label{Thirring-2}
\end{equation}
This result will be useful in deriving the explicit expression of the critical Josephson current.

%% file: sections/large_N_formalism.tex
\section{Large-$N$ formalism}
In order to implement in a mathematically controlled way the large-$N$ limit of the quasi-spin formulation of the BCS model 
introduced in the previous Section, one can resort to standard techniques developed in general for 
quantum systems with infinitely many degrees of freedom, based on the so-called 
Gelfand, Neimark and Segal (GNS) construction~\cite{Bratteli}.

Quasi-spin systems are typical examples of quasi-local $C^*$ algebras; given any state $\Omega$
on one such algebra $\mathcal{A}$, {\it i.e.} a functional from the elements of the algebra to the field of complex numbers $\C$,
one can associate to it a representation on a Hilbert space endowed with a cyclic vector 
that reproduces all expectation values relative to the given state functional. In more precise terms,
given a state $\Omega$ over $\mathcal{A}$, there always exists a cyclic representation $\pi_\Omega$ of $\mathcal{A}$
over an Hilbert space $\mathcal{H}_\Omega$, with a state vector $\Psi_\Omega$, such that
$\Omega(X) = \braket{\Psi_\Omegastate|\pi_\Omegastate{(X)}|\Psi_\Omegastate}$ for all elements $X\in\mathcal{A}$.
Moreover, the representation is unique up to unitary equivalence.

As a consequence of this GNS construction, generic state vectors in $\mathcal{H}_\Omegastate$ 
can be obtained by acting with a suitable represented operator 
$\pi_\Omegastate(X)$ on the cyclic vector $\vert\Psi_\Omegastate\rangle$. Furthermore, any matrix element in the representation space $\mathcal{H}_\Omegastate$ corresponds to an expectation value with respect to $\Omegastate$ in the abstract formulation. 
Indeed, if we let $\ket{\psi_A} = \pi_\Omegastate(A)\ket{\Psi_\Omegastate}$ and $\ket{\psi_B} = \pi_\Omegastate(B)\ket{\Psi_\Omegastate}$ be arbitrary state vectors in $\mathcal{H}_\Omegastate$, then, for all $X\in\algebra$,
\begin{equation*}
\label{ch1-eq:weak-oper-top}
\begin{split}
\braket{\psi_A|\pi_\Omegastate(X)|\psi_B} 
&=\braket{\Psi_\Omegastate|\pi_\Omegastate\dag(A)\pi_\Omegastate(X)\pi_\Omegastate(B)|\Psi_\Omegastate}  \\
& = \braket{\Psi_\Omegastate|\pi_\Omegastate(A\star XB)|\Psi_\Omegastate}\ .
\end{split}
\end{equation*}

When dealing with density matrices, the GNS construction is achieved via purification of the state.
For simplicity, let us illustrate the procedure using the example of a finite-dimensional system with Hilbert space $\C^d$,
although the construction is general. 
The set of observables for such a system is that of $d\times d$ complex matrices, $\matrixset{(d,\C)}$. 
Any density matrix $\rho$ describing a state of the system can be given in its spectral decomposition,
\begin{equation}
\label{ch1-eq:density-matrix-spectral-decomposition}
 \rho = \sum_{i=1}^d r_i\ket{r_i}\bra{r_i}\ ,\quad r_i\ge0\ ,\quad \sum_{i=1}^dr_i = 1\ , \quad\braket{r_i|r_j} =\delta_{ij}\ ,
\end{equation}
in terms of the basis formed by its eigenvectors $\ket{r_i}$. Let us now double the Hilbert space, and consider in it
the vector $\ket{\Psi_\rho}\in\C^d\otimes\C^d$ defined by
\begin{equation}
\label{ch1-eq:purified-vector}
\ket{\Psi_\rho} = \sum_{i=1}^d\sqrt{r_i}\ket{r_i}\otimes\ket{r_i}\ .
\end{equation}
Then, it turns out that all expectation values of the system observables can be computed 
as scalar products in $\C^d\otimes\C^d$:
\begin{equation}
\label{ch1-eq:purif-exp-value}
\langle X\rangle = \braket{\Psi_\rho|X\otimes\identity|\Psi_\rho}\ ,\quad\forall X\in\matrixset{(d,\C)}\ ,
\end{equation}
where $\identity$ represents the identity operator.
The map $\pi\colon X\mapsto X\otimes\identity$ provides a \emph{representation} of $\matrixset{(d,\C)}$ on the Hilbert space $\C^d\otimes\C^d$ as $\representation{(\matrixset{(d,\C)})}=\matrixset{(d,\C)})\otimes\identity\subset\matrixset{(d^2,\C)}$. 
The advantage of this construction is that the density matrix $\rho$ is replaced by the pure state vector $\ket{\Psi_\rho}$.  
Such a procedure is called \emph{purification}.

Notice that the map $\representationprimed\colon X\mapsto \identity \otimes X$ 
also provides a representation of $\matrixset{(d,\C)})$, 
$\representationprimed{(\matrixset{(d,\C)})}=\identity\otimes\matrixset{(d,\C)})$: 
it commutes with $\representation{(\matrixset{(d,\C)})}$ and is thus called its \emph{commutant}.

Given a Hamiltonian $H$ for a $d$-levels system, with spectral decomposition
\begin{equation}
    H = \sum_{i=1}^d \eta_i\ket{i}\bra{i}\ ,
\end{equation}
let $\rho_\beta$ be the Gibbs state at inverse temperature $\beta$, with spectral decomposition
\begin{equation}
\rho_\beta = \sum_{i=1}^d \rho_i\ket{i}\bra{i}\ ,\qquad \rho_i=\frac{e^{-\beta \eta_i}}{\sum_{i=1}^de^{-\beta \eta_i}}\ ;
\end{equation}
then, the corresponding purified state is given by:
\begin{equation*}
\ket{\Psi_\beta} = \sum_{i=1}^d\sqrt{\rho_i}\ket{i}\otimes\ket{i}\ .
\end{equation*}
Unlike $\rho_\beta$ which is invariant under the dynamics generated by $H$, the GNS vector $\vert\Psi_\beta\rangle$ is not. 
Indeed, one finds: 
\begin{equation*}
    e^{-itH\otimes 1}\ket{\Psi_\beta} = \sum_{i=1}^d \sqrt{\rho_i}\, (e^{-it\eta_i}\ket{i})\otimes\ket{i}.
\end{equation*}
In order to make it time-invariant, one needs to find a suitable representation of the dynamics: 
this is given by renormalizing the Hamiltonian through the subtraction from $H\otimes 1$ 
of a contribution from the commutant: $H\otimes 1\mapsto H\otimes 1-1\otimes H$. Then
\begin{equation*}
    e^{-it(H\otimes1-1\otimes H)}\ket{\Psi_\beta} = \sum_{i=1}^d \sqrt{r_i} (e^{-it\eta_i}\ket{i})\otimes(e^{it\eta_i}\ket{i})=\ket{\Psi_\beta}.
\end{equation*}

In the case of the BCS quasi-spin formulation introduced in the previous Section, 
the representation space is $\C^{2^{2N}}$ and, in the basis $\{\ket{s,s_z}\}$ of eigenstates 
of the collective Casimir operator $(\vec{S}^N)^2$ and spin component $S^N_z$ in (\ref{collective-spin}), 
the purified state associated with Gibbs state $\rho_\beta^N$ 
in (\ref{Gibbs-state}) can be expressed as:
\begin{equation}
\label{weights}
\ket{\Omega_\beta^N} = \sum_{s=0}^{N/2}\sum_{s_z=-s}^s \sqrt{\rho_\beta^N(s,s_z)}\,\bigoplus_{\alpha=1}^{d(s)}\,\ket{s,s_z}_\alpha\otimes\ket{s,s_z}_\alpha\ ,
\end{equation}
with Boltzmann weights:
\begin{equation}
\label{BoltzmannWeightsSC}
\rho_\beta^N(s,s_z)=\,\frac{e^{-\beta \eta^N(s,s_z)}}{\sum_{s'=0}^{N/2}\,d(s')\,\sum_{s_z'=-s'}^{s'}e^{-\beta \eta^N(s',s_z')}}\ ,
\end{equation}
where $\eta^N$ are the eigenvalues of $H^N$ in (\ref{BCS-spin-2}),
\begin{equation}
\eta^N(s,s_z)= -2\epsilon s_z-\frac{2T_C}{N}\bigl(s(s+1)-s_z(s_z-1)\bigr)\ ,
\label{spectrum}
\end{equation}
while $d(s)$ denotes the multiplicity of the irreducible representation with orthonormal basis 
$\ket{s,s_z}_\alpha$, $1\leq\alpha\leq d(s)$ \cite{Thirring2,Thirring3}.
Then, the expectation value of any spin-observable $X$ in the Gibbs state $\rho^N_\beta$ 
can be expressed as the following average:
\begin{equation}
\langle X \rangle_\beta^N = \braket{\Omega_\beta^N | X_N\otimes\identity \,| \Omega_\beta^N}\ .
\end{equation}
In the following, we will always work within the GNS representation.

%% file: sections/algebra_fluctuations.tex
\section{Algebra of fluctuations}

As mentioned in the Introduction, we are interested in selecting collective quasi-spin operators having a well-defined limit as $N$ becomes large. 
Specifically, we will identify collective $N$-spin operators $p^N$ and $E_\pm^N$ 
which, in the large-$N$ limit, become the momentum operators, $p_\varphi$, and angular exponentials, $e^{\pm i\varphi}$,
respectively, obeying the commutation relations of the Heisenberg algebra on the circle
given in (\ref{weyl2}) (for simplicity, carets will be henceforth omitted).

More precisely, the momentum variable is defined by means of a collective spin component in the $z$ direction 
renormalized by a contribution from the commutant, without any rescaling with powers of $N$:
\begin{equation}
p^N \equiv S^N_z\otimes\identity-\identity\otimes S^N_z\ .
\label{momentum}
\end{equation}
The reason for the renormalization is that in this way $p^N$ annihilates the Gibbs state $|\Omega_\beta^N\rangle$
introduced in the previous Section:
\begin{equation}
p^N\ket{\omeganb} = \sum_{s=0}^{N/2}\sum_{s_z=-s}^s \hskip-.2cm \sqrt{\rho^N_\beta(s,s_z)}\, 
\bigoplus_{\alpha=1}^{d(s)}(s_z-s_z)\ket{s,s_z}_\alpha\otimes\ket{s,s_z}_\alpha = 0\ .
\label{on-vacuum}  
\end{equation}
The relevant part of the spectrum%
\footnote{The other part of the spectrum of $p^N$ is associated with its component in the commutant algebra.} 
of $p^N$ is then retrieved by acting on $|\Omega_\beta^N\rangle$ with $S^N_\pm\otimes\identity$. 
Indeed, the $su(2)$ commutation relation between $S^N_z$ and $S^N_+$ leads to
\begin{equation}
\Big[ p^N , (S^N_\pm)^n\otimes\identity\Big]= \pm n\, \big[(S^N_\pm)^n\otimes\identity\big]\ .
\end{equation}
Then, taking into account that, by construction, $(S^N_\pm)^{N+1}=0$, one gets
that the vectors $\big[(S^N_\pm)^n \otimes\identity\big] \ket{\omeganb}$, for $0\leq n\leq N$, are eigenvectors of $p^N$:
\begin{equation}
p^N\,  \big[(S^N_\pm)^n \otimes\identity\big] \ket{\omeganb} = \pm n\, \big[(S^N_\pm)^n \otimes\identity\big] \ket{\omeganb}\ .
\end{equation}
Thus, the thermal state $\ket{\omeganb}$ constitutes the vacuum in the GNS representation and $p^N$ can be interpreted
as the operator counting the number of quasi-spin excitations above it.

As ``conjugate'' operator to $p^N$, it is convenient to introduce the rescaled quantities:
\begin{equation}
E^N_\pm\equiv \frac{S^N_\pm\otimes\identity}{\mathfrak{c} N}\ ,
\label{phase}
\end{equation}
with $\mathfrak{c}$ as in (\ref{Thirring-2}). These operators are not unitary:
\begin{equation*}
E^N_- = (E_+^N)\dag\ ,\quad    E_+^N\,E_-^N = \frac{S^N_+S^N_-\otimes \identity}{(\mathfrak{c} N)^2}\ne\identity\ ;
\end{equation*}
nevertheless, $p^N$ and $E_\pm^N$ satisfy the characteristic algebraic relation for the Heisenberg algebra on the circle:
\begin{equation}
\label{algebra-N}
\Big[p^N, E_\pm^N\Big] = \pm\,E_\pm^N.
\end{equation}
The collective operators $p^N$ and $E_\pm^N$ are examples of quantum fluctuations.
Notice that, in order to obtain the Heisenberg algebra on the circle in the large-$N$ limit, as discussed in detail in the following, 
only $E_\pm^N$ need to be rescaled with $N$, 
while no such rescaling is necessary for $p^N$; due to this ``asymmetry'', these operators are sometimes
referred to as {\it abnormal fluctuations} \cite{Verbeure}.

We want now to study the large-$N$ behaviour of correlation functions 
involving $p^N$ and $E_\pm^N$. In order to simplify the notation, 
we denote expectation values in the GNS cyclic vector $\ket{\omeganb}$ as:
\begin{equation*}
\expvalue{\,\cdot\,}^N_\beta \equiv \braket{\omeganb | \,\cdot\,| \omeganb}\ .
\end{equation*}
Let $\{\alpha_j\in\R,\ m_j,n_j\in\N\}$, $j=1,2,\ldots, r$, be a collection of (possibly null) constants;
we will focus upon correlation functions of the form:
\begin{equation}
\label{F-correlation}
F_\beta^N\big(\{\alpha_j\},\{m_j\},\{n_j\}\big) \equiv 
\Bigexpvalue{\prod_{j=1}^r e^{i\alpha_j p^N} (E_-^N)^{n_j} \, (E_+^N)^{m_j}}_\beta^N\ .
\end{equation}
\begin{remark}

Under the norm closure of their algebra, polynomials  in the operators $(E^N_\pm)^m$ and $e^{i\alpha p^N}$ generate a $C^*$  algebra. 
Given a state on such an algebra, a dense subspace of vectors in the associate Hilbert space of the corresponding 
GNS-representation is built with the action of any such product on the GNS cyclic vector. 
By using the relations \eqref{algebra-N}, one can then express any scalar product 
in such Hilbert space as a suitable correlation function
$F_\beta^N\big(\{\alpha_j\},\{m_j\},\{n_j\}\big)$.
\end{remark}

The following result will then give a precise meaning to the large-$N$ limit representation of the fluctuation
operators $p^N$ and $E_\pm^N$ in terms of momentum $\hat{p}_\varphi$ and angle variable exponentials $e^{\pm i\hat{\varphi}}$, respectively;
it defines correlations at the {\it mesoscopic} scale, in between the microscopic and the classical, macroscopic ones.  

\begin{restatable}[Mesoscopic correlation functions]{theorem}{reconstruction}
\label{ch5.1-th:reconstruction}
The large-$N$ limit of the correlation functions in~\eqref{F-correlation} exists and defines 
a state functional ${\Omega}$ on the Heisenberg algebra on the circle $\mathcal{C}$ such that
\begin{align}
\nonumber
 \lim_{N\to\infty} F_\beta^N\big(\{\alpha_j\},\{m_j\},\{n_j\}\big) & =
\Omega\Biggl(\prod_{j=1}^r e^{i\alpha_j  \hat{p}_\varphi} e^{i(m_j-n_j)\hat{\varphi}}\Biggr)\\
& =\delta_{m,n}\ e^{i\sum_{j=1}^r\sum_{k=1}^j\alpha_{k}(m_j-n_j)}\ ,
\label{mesoscopic-correlations}
\end{align}
where
$m \equiv \sum_{j=1}^r m_j$ and $n \equiv\sum_{j=1}^r n_j$. 
\end{restatable}
\vskip -.3cm
\noindent
The proof of this result can be found in Appendix A.

The Heisenberg algebra on the circle can be represented on the Hilbert space $L^2\bigl([0,2\pi]\bigr)$ as discussed  at the end of Section~\ref{sec:chargequbit}.
One thus finds that the Hilbert space scalar product~\eqref{scalarproduct}  is precisely what one obtains from (\ref{mesoscopic-correlations}) by setting $\alpha_j=0$, $\forall j$:
$$
\lim_{N\to\infty} \Bigexpvalue{\prod_j (E_-^N)^{n_j} (E_+^N)^{m_j}}_\beta^N=\Omega\bigl(e^{i(m-n)\hat{\varphi}}\bigr)=\delta_{m,n}.
$$
Notice that $p^N\ket{\Omega}_\beta^N=0$ yields $F_\beta^N\big(\{\alpha_j\},\{m_j=0\},\{n_j=0\}\big)=1$; therefore, the state functional $\Omega$ simply corresponds to taking expectation values with respect to the vacuum $|0\rangle$,
and we can simply write: $\Omega(\,\cdot\,) = \braket{0 | \,\cdot\, | 0}$.
In addition, 
notice that in the large-$N$ limit the ordering of the powers of $E_\pm^N$ is irrelevant
\begin{equation*}
\lim_{N\to\infty} \Bigexpvalue{\prod_j (E_-^N)^{n_j} (E_+^N)^{m_j}}^N_\beta =
\lim_{N\to\infty} \bigexpvalue{ (E_-^N)^n\, (E_+^N)^m\,  }^N_\beta = \delta_{m,n}\ .
\end{equation*}
Theorem~\ref{ch5.1-th:reconstruction} establishes the precise meaning of the mesoscopic limit 
that leads from the microscopic algebra of finite size spin 
observables to the mesoscopic algebra on the circle. In other terms, in the  large-$N$ limit, 
the correlation functions involving $e^{i\alpha p^N}$ and $E^N_\pm$ 
can be represented as expectations of $e^{i\alpha \hat{p}_\varphi}$ and $e^{im\hat{\varphi}}$, 
with respect to the pure state $\ket{0}\in L^2([0,2\pi])$.

\begin{definition}[Mesoscopic limit]\label{ch5.1-def:mesoscopic-lim}
Let $\{X^{(N)}, N\in\N\}$ be a sequence of observables of the quasi-spin system, 
each one given by a linear combination of powers of $e^{i  p^N}$ and $E_\pm^N$. 
We say that it converges in the \emph{mesoscopic limit} to $X\in\mathcal{C}$ if, for all $n,\,n',\,m,\,m'\in\N$,
\begin{equation}
\label{ch5.1-eq:def-fluc-lim}
\lim_{N\to\infty} 
\bigexpvalue{(E_-^N)^{n}(E_+^N)^{n'} X^{(N)} (E_+^N)^{m} (E_-^N)^{m'}}^N_\beta
 =\Omega(e^{-i(n-n')\hat{\varphi}}\,X\,e^{i(m-m')\hat{\varphi}})\ ,
\end{equation}
and simply write
\begin{equation}
\mlim_{N\to\infty} X^{(N)} = X.
\end{equation}
\end{definition}

\begin{remark}
\label{rem:corr2}
The left hand side of~\eqref{ch5.1-eq:def-fluc-lim} has always the form of a linear combination of correlation functions such as~\eqref{F-correlation}, being $X^{(N)}$ a linear combination of powers of $e^{i p^N}$, $E_\pm^N$.
The right hand side is instead the matrix element of $X$ with respect to the vectors $\ket{n}$ and $\ket{m}$, 
with generic $n,\ m$, thus allowing the complete reconstruction of the operator $X$.
\end{remark}

%% file: sections/dynamics_fluctuations.tex
\section{Dynamics of fluctuations}
\label{ch5.1-dyn}

We now turn our attention to the emerging dynamics of the collective fluctuations at the mesoscopic scale. 
We will first study the fate of the microscopic Hamiltonian of a single superconductor as $N$ becomes large,
and the dynamics it induces on the Heisenberg algebra on the circle. Then we shall focus on a system 
made of two superconductors, coupled through a suitable tunnelling term, in order to properly model the dynamics of
a Josephson junction.

\subsection{Single superconductor}

For sake of generality, we shall allow for the presence of a non-vanishing chemical potential by adding 
a term proportional to $S^N_z$ to the strong-coupling BCS Hamiltonian (\ref{BCS-spin-2}):
\begin{equation}
\mathcal{H}^N \equiv H^N +\mu\sum_{k=1}^N \sigma_z^{(k)}=-2(\epsilon-\mu) S^N_z-\frac{2 T_c}{ N}S^N_+S^N_-\ ,
\end{equation}
so that for every Cooper pair added to the system, there is an additional energy shift $-2\mu$.
As $[H^N\,,\,S^N_z]=0$, the spectrum of $\mathcal{H}^N$ is simply a shift of that of $H^N$ given in (\ref{spectrum}):
\begin{equation*}
\mathcal{H}^N\ket{s,s_z} = \big(\eta^N(s,s_z)+2\mu s_z\big)\ket{s,s_z}\ .
\end{equation*}

As explained in Section 4, the generator of the time evolution in the GNS representation is 
not simply given by $\mathcal{H}^N\otimes\identity$, but by its renormalization through the subtraction 
of a contribution from the commutant, 
$\mathcal{H}^N\otimes\identity - \identity\otimes\mathcal{H}^N=H^N\otimes\identity - \identity\otimes H^N +2\mu p^N$;
this redefinition does not affect the dynamics of any quasi-spin observable, as they are of the generic form $X\otimes\identity$.

At finite $N$, the time-evolution operator is then given by
\begin{equation}
\mathcal{U}^N(t)\equiv e^{-it(\mathcal{H}^N\otimes\identity-\identity\otimes \mathcal{H}^N)} =
e^{-it({H}^N\otimes\identity-\identity\otimes {H}^N)}\, e^{-it\, (2\mu  p^N)}\ ,
\label{time-evolution}
\end{equation}
and has the advantage of leaving invariant the GNS state vector:
\begin{equation}
\mathcal{U}^N(t)\ket{\omeganb} = \ket{\omeganb}\ .
\end{equation}
In the large-$N$ limit however, only the last exponential in (\ref{time-evolution}) survives,
as it turns out that the strong-coupling BCS Hamiltonian $H^N$ does not contribute to the time evolution
of fluctuation operators. Indeed, 
\begin{restatable}[Mesoscopic dynamics]{theorem}{dynamics}
\label{ch5.1-th:dyn}
The evolution operator \eqref{time-evolution} converges in the mesoscopic limit to
\begin{equation}
\mlim_{N\to\infty} \mathcal{U}^N(t) = \mathcal{U}(t) \equiv e^{-it\, \mathcal{H} } = e^{-it (2\mu  \hat{p}_\varphi)}\ .
\end{equation}
\end{restatable}
\noindent
The Hamiltonian $\mathcal{H}$ is the generator of time-translations over the circle (for the proof, see the Appendix A).

\subsection{Josephson junction}

As a Josephson junction involves two independent superconducting layers, in order to cope with this new physical situation
we need to extend the previous description. In practice, we have to ``double'' the construction so far used and
in doing so we shall use the label $L$ and $R$ to distinguish quantities 
referring to either one of the two superconductors.
The collective spin operators, now acting on the Hilbert space $\mathbb{C}^{2^N}\otimes\mathbb{C}^{2^N}$, 
will be denoted by $\mathcal{S}^N_i$ and $\mathcal{T}^N_i$, $i=x,y,z$, for the $L$ and $R$ electrode, respectively:
\begin{eqnarray}
\label{spin-two-1}
&& \mathcal{S}^N_i\equiv  S^N_i \otimes\identity_R\ ,\qquad \mathcal{S}^N_\pm\defsym \mathcal{S}^N_x\pm i\mathcal{S}^N_y\ ,\\
\label{spin-two-2}
&& \mathcal{T}^N_i\equiv \identity_L\otimes S^N_i\ , \qquad \mathcal{T}^N_\pm\defsym \mathcal{T}^N_x\pm i\mathcal{T}^N_y\ ,
\end{eqnarray}
with $S^N_i$ as in (\ref{collective-spin}) and $\identity_R$, $\identity_L$ the identity operator in the $R$, $L$
subspaces; these operators separately satisfy the commutations relations of two independent $su(2)$ algebras. 
The common eigenstates of the Casimir $(\vec{\mathcal{S}}^N)^2$, $(\vec{\mathcal{T}}^N)^2$, 
and $\mathcal{S}^N_z$, $\mathcal{T}^N_z$ operators
are tensor products $\ket{s,s_z}\otimes\ket{t,t_z}$ of the two, single $su(2)$ basis vectors representations, 
with $s,t$ ranging from 1 to $N/2$ and $s_z=0,\pm1,\cdots\pm s$, $t_z = 0,\pm1,\cdots\pm t$.

With these quasi-spin operators one construct two strong-coupling BCS Hamiltonians,
\begin{equation}
\label{two-layers}
\mathcal{H}_L^N = -2\epsilon_L \mathcal{S}^N_z-\frac{2\, T_c^L}{N} \mathcal{S}^N_+\mathcal{S}^N_-\ ,\qquad
\mathcal{H}_R^N = -2\epsilon_R \mathcal{T}^N_z-\frac{2\, T_c^R}{N} \mathcal{T}^N_+\mathcal{T}^N_-\ ,
\end{equation}
where, for generality, we assumed different energy parameters $\epsilon_L$, $\epsilon_R$, and critical temperatures
$T_c^L$, $T_c^R$ in the two layers.

Let us indicate with $\rho_\beta^N$ the Gibbs state associated with 
the total system Hamiltonian $\mathcal{H}_L^N + \mathcal{H}_R^N$,
assuming no interactions between the two layers and a common inverse temperature $\beta$; it reduces to the tensor product
of the corresponding Gibbs states pertaining to the two $L$- and $R$-subsystem,
$\rho_\beta^N=\rho_{\beta,L}^N \otimes \rho_{\beta,R}^N$,
so that, given two observables $X^N_L$, $X^N_R$, 
acting separately on the two layers, their thermal expectation values factorizes:
\begin{equation}
{\rm Tr}\big[ \rho_\beta^N\, X^N_L \otimes X^N_R\big]= {\rm Tr}\big[ \rho_{\beta,L}^N\, X^N_L \big]\
{\rm Tr}\big[ \rho_{\beta,R}^N\, X^N_R \big]\ .
\end{equation}
The corresponding GNS cyclic vector $\ket{\omeganb}$ will then be given by 
\begin{align}
\nonumber
\ket{\omeganb} &=
\sum_{s,s_z}\sum_{t,t_z}\sqrt{\rho^N_\beta(s,s_z;t,t_z)}\bigoplus_{\alpha_L=1}^{d(s)}
\bigoplus_{\alpha_R=1}^{d(t)}\bigl(\ket{s,s_z}_{\alpha_L}\otimes\ket{s,s_z}_{\alpha_L}\bigr)\\
    \label{GNS-vacuum}
&\hskip 7cm \otimes\bigl(\ket{t,t_z}_{\alpha_R}\otimes\ket{t,t_z}_{\alpha_R}\bigr)\ ,
\end{align}
with $\rho^N_\beta(s,s_z;t,t_z) = \rho^N_{\beta,L}(s,s_z)\, \rho^N_{\beta,R}(t,t_z)$,
where the r.h.s. quantities are defined as in (\ref{BoltzmannWeightsSC}) for both the $L$ and $R$ part,
while $d(s)$, $d(t)$ take into account as before the multiplicity of the various $su(2)$ spin-representations.

Having to deal with two superconducting layers, the fluctuation operators in the GNS representation will
be obtained generalizing the construction of Section 5 relative to a single superconductor. 
Recalling (\ref{momentum}) and (\ref{phase}), we shall then use the following definition 
for the momentum and its conjugate phase operators in the $L$ and $R$ layers:
\begin{eqnarray}
\label{fluctuations-def-1}
&& p_L^N = \Bigl(S^N_z\otimes\identity_L - \identity_L\otimes S^N_z\Bigr)\otimes\identity_R\otimes\identity_R\ ,\\
\label{fluctuations-def-3}
&& (E^N_{L,\pm})^m = \left(\frac{S^N_\pm}{\mathfrak{c}_L N}\right)^m\! \otimes\identity_L\otimes\identity_R\otimes\identity_R\ , \\
&& p_R^N = \identity_L\otimes\identity_L\otimes\Bigl(S^N_z\otimes\identity_R - \identity_R\otimes S^N_z\Bigr),\\
\label{fluctuations-def-2}
&& (E^N_{R,\pm})^m = \identity_L\otimes\identity_L\otimes\left(\frac{S^N_\pm}{\mathfrak{c}_R N}\right)^m\!\otimes\identity_R\ ,
\qquad m\in\mathbb{N}\ ,
\end{eqnarray}
with $S^N_i$, $i=z,\pm$, as in (\ref{collective-spin}), and  $\mathfrak{c}_L$, $\mathfrak{c}_R$
connected to the ``gaps'' in the two layers as in (\ref{Thirring-2}).
Since the Hilbert space for the system of two layers is 
$\mathbb{C}^{2^N}\otimes\mathbb{C}^{2^N}$, its GNS extension will involve the tensor products of four factors. 
The choice of the ordering of these four spaces is arbitrary, as all these possible choices are unitarily related; 
the ordering adopted in 
\hbox{(\ref{fluctuations-def-1})-(\ref{fluctuations-def-2})} is the most convenient for what follows.

One can then check that the analogue of the algebraic relation \eqref{algebra-N} 
hold for the two separate subsystems:
\begin{eqnarray}
\label{L-algebra}
 && \Big[p_L^N,\, E^N_{L,\pm}\Big] = \pm E^N_{L,\pm}\ ,\qquad \Big[p_R^N,\, E^N_{R,\pm}\Big] =\pm E^N_{R,\pm}.
\end{eqnarray}
Since the operators of one quasi-spin system commute with those of the other one, 
and the GNS state $\ket{\omeganb}$ in (\ref{GNS-vacuum}) does not carry correlations between the two quasi-spin systems, 
the whole analysis of Section 5 can be carried over to the new setting. 
The new target algebra of fluctuations is the tensor product of two commuting algebras 
on the circle, $\mathcal{C}^L\otimes\mathcal{C}^R$. 
With obvious adaptation and extension of the notation in~\eqref{F-correlation}, correlation functions as
\begin{eqnarray}
\nonumber
&& \mathcal{F}_\beta^N\big(\{\alpha\},\{m\},\{n\}; \{\alpha'\},\{m'\},\{n'\}\big)\\
\nonumber
&& \hskip 2cm
\equiv \Bigexpvalue{
        \prod_{j=1}^r e^{i\alpha_j p_L^N} (E^N_{L,-})^{n_j} (E^N_{L,+})^{m_j} 
        \prod_{k=1}^{r'} e^{i\alpha'_k p_R^N} (E^N_{R,-})^{n'_{k}} (E^N_{R,+})^{m'_{k}}}^N_\beta\ ,
\end{eqnarray}
factorize into products of correlation functions of the subsystems:
\begin{eqnarray*}
&& \mathcal{F}_\beta^N\big(\{\alpha\},\{m\},\{n\}; \{\alpha'\},\{m'\},\{n'\}\big) = 
\Bigexpvalue{\prod_{j=1}^r e^{i\alpha_j p_L^N} (E^N_{L,-})^{n_j} (E^N_{L,+})^{m_j}}^N_\beta\\
&&\hskip 8cm
       \times\ \Bigexpvalue{ \prod_{k=1}^{r'} e^{i\alpha'_k p_R^N} (E^N_{R,-})^{n'_{k}} (E^N_{R,+})^{m'_{k}}}^N_\beta\ .
\end{eqnarray*}
Taking the large-$N$ limit, we can apply Theorem~\ref{ch5.1-th:reconstruction} to both factors. 
The above function is then reconstructed as a correlation function on the tensor product of two 
commuting algebra on the circle, $\mathcal{C}_L\otimes\mathcal{C}_R$, with respect to a factorized state:
\begin{eqnarray}
\nonumber
&&\lim_{N\to\infty} \mathcal{F}_\beta^N\big(\{\alpha\},\{m\},\{n\}; \{\alpha'\},\{m'\},\{n'\}\big)\\
\nonumber
&& \hskip 3cm =\Omega\Bigl(\prod_{j=1}^r e^{i\alpha_j \hat{p}_{\varphi_L}} e^{i(m_j-n_j)\hat{\varphi}_L}\prod_{k=1}^{r'} e^{i\alpha_k' \hat{p}_{\varphi_R}} 
e^{i(m_k'-n_k')\hat{\varphi}_R}\Bigr) \\
\nonumber
&& \hskip 3cm = {\Omega}_L\Bigl(\prod_{j=1}^r e^{i\alpha_j \hat{p}_{\varphi_L}} e^{i(m_j-n_j)\hat{\varphi}_L}\Bigr)\
	{\Omega}_R \Bigl(\prod_{k=1}^{r'} e^{i\alpha_k' \hat{p}_{\varphi_R}} e^{i(m_k'-n_k')\hat{\varphi}_R}\Bigr)\\
\nonumber
&& \hskip 3cm = (e^{i\chi}\, \delta_{m,n})\,(e^{i\chi'}\,\delta_{m',n'})\ ,
\end{eqnarray}
where the phase factors $\chi$ and $\chi'$ are given by
\begin{equation*}
	\chi = \exp{\Bigl(i\sum_{j=1}^r\sum_{k=1}^j\alpha_{k}(m_j-n_j)\Bigr)},\quad \chi' = \exp{\Bigl(i\sum_{k=1}^{r'}\sum_{k=1}^j\alpha'_{k}(m'_j-n'_j)\Bigr)}\ ,
\end{equation*}
and $m =\sum_{j=1}^r m_j$, $n =\sum_{j=1}^r n_j$, $m' = \sum_{j=1}^{r'} m'_j$, $n' =\sum_{j=1}^{r'} n'_j$.
We can represent $\mathcal{C}_L\otimes\mathcal{C}_R$ onto $L^2\bigl([0,2\pi]\bigr)\otimes L^2\bigl([0,2\pi]\bigr)$, 
with the state $\Omega$ being the expectation on the vacuum state, realized by the constant function:
\begin{equation}
\label{circle-vacuum}
\braket{\vec\theta|0}=\braket{\theta_L |0 }_L\, \braket{\theta_R | 0}_R = \frac{1}{2\pi}\ ,
\end{equation}
where $\vec\theta=(\theta_L,\, \theta_R)$.
In the following, for sake of compactness, we shall always use the notation
$ \expvalue{\,\cdot\,} = \braket{0 | \, \cdot\, |0}$ in order to indicate the expectation on $\mathcal{C}_L\otimes\mathcal{C}_R$.

Finally, the definition of the mesoscopic limit of junction observables, 
analogue to Definition~\ref{ch5.1-def:mesoscopic-lim} for the single BCS superconductor, reads as follows:

\begin{definition}[Mesoscopic limit]
\label{def:mesoscopic-lim-double}
Let $\{X^{(N)},\ N\in\N\}$ be a sequence of observables of the double quasi-spin system, 
each one given by a linear combination of powers of $e^{i  p_L^N}$, $E^N_{L,\pm}$ and
$e^{i  p_R^N}$, $E^N_{R,\pm}$. For any $m_L,\, m_L',\, m_R,\, m_R' \in\N$, let us define the operator
\begin{equation}
G^N(m_L,m_L',m_R,m_R') = (E^N_{L,-})^{m_L} (E^N_{L,+})^{m_L'} (E^N_{R,-})^{m_R} (E^N_{R,+})^{m_R'} \ .
\end{equation}
Then we say that the sequence converges in the \emph{mesoscopic limit} to $X\in\mathcal{C}_L\otimes\mathcal{C}_R$,
and write $\mlim_{N\to\infty} X^{(N)} = X$ if
\begin{multline}
\lim_{N\to\infty} \Bigexpvalue {\bigl(G^N(n_L,n_L',n_R,n_R')\bigr)\dag \ X^{(N)}\ G^N(m_L,m_L',m_R,m_R')}^N_\beta = \\ 
= \bigexpvalue{e^{-i(n_L-n_L')\hat{\varphi}_L}e^{-i(n_R-n_R')\hat{\varphi}_R}\ X\ e^{i(m_R-m_R')\hat{\varphi}_R} e^{i(m_L-m_L')\hat{\varphi}_L}},
\end{multline}
for any $n_L,\, n_L',\, n_R,\, n_R',\, m_L,\, m_L',\, m_R,\, m_R' \in\N$.
\end{definition}

\subsection{Charge qubit}

Using the techniques and definitions introduced above, we can now discuss the mesoscopic dynamics of the
quantum circuit relative to the charge qubit, consisting of a Josephson junction together with an external 
voltage and a capacitor as discussed in Section 2.

The sum of the two Hamiltonians in (\ref{two-layers}) generates the dynamics of two independent superconducting layers.
However, in order to describe the tunneling of Cooper pairs in a Josephson junction, an interaction between the
two layers is needed. Such a coupling can be conveniently described by a bilinear Hamiltonian written in terms
of the collective spin operators $\mathcal{S}_i$ and $\mathcal{T}_i$ introduced in (\ref{spin-two-1}),
(\ref{spin-two-2}) as:
\begin{equation}
\nonumber
\widehat H^N =  \frac{\lambda}{N^2}(\mathcal{S}^N_+\mathcal{T}^N_-+\mathcal{S}^N_-\mathcal{T}^N_+)\ ,
\end{equation}
where $\lambda$ is a suitable coupling constant and the scaling $N^{-2}$ takes into account 
that the passage of Cooper pairs in the junction is a surface effect.

As we shall be working in the GNS representation, we need to extend both the ``free'' double-layer Hamiltonian,
$H_L^N + H_R^N$, and the tunneling interaction $\widehat H^N$ above, as operators acting
on the Hilbert space obtained by doubling $\mathbb{C}^{2^N}\otimes\mathbb{C}^{2^N}$,
the two layers Hilbert space. Using the same ordering of the four tensor factors in this extended space
already introduced in the definitions 
\hbox{(\ref{fluctuations-def-1})-(\ref{fluctuations-def-2})}, one can conveniently
define the ``free'' Hamiltonian as:
\begin{eqnarray}
\nonumber
&& H^N_{\mathrm{free}} \equiv \left(H^N_L\otimes\identity_L-\identity_L\otimes H^N_L\right)\otimes\identity_R\otimes\identity_R\\  
&& \hskip 3cm+\identity_L\otimes\identity_L\otimes\left(H^N_R\otimes\identity_R-\identity_R\otimes H^N_R\right)\ ,
\label{H-free}
\end{eqnarray}
where $H^N_L$, $H^N_R$ are exactly as in (\ref{BCS-spin-2}), but with the parameters $\varepsilon$ and $T_c$
replaced by $\varepsilon_L$, $T_c^L$ and $\varepsilon_R$, $T_c^R$, respectively; one easily checks that $H^N_{\mathrm{free}}$
annihilates the GNS vacuum state $\ket{\omeganb}$ given in (\ref{GNS-vacuum}). Similarly, the tunnelling Hamiltonian
$\widehat H^N$ can be simply rewritten in the GNS representation in terms of the fluctuation operators
introduced in (\ref{fluctuations-def-3}) and (\ref{fluctuations-def-2}) as:
\begin{equation}
H_\mathrm{int}^N \equiv \lambda \mathfrak{c}_L \mathfrak{c}_R \Big(E^N_{L,+}\, E^N_{R,-} + E^N_{L,-}\, E^N_{R,+}\Big)\ .
\label{interaction}
\end{equation}

In order to describe the dynamics of a charge qubit, we need to add a capacitive term
to the above Hamiltonian pieces, keeping track of the excess charge induced by to the tunneling
of Cooper pairs and the presence of an external gate voltage. 
Recall that excitations above the thermal state are created 
by acting with the fluctuation operators $E^N_{L, \pm}$ and $E^N_{R,\pm}$ on the the GNS vacuum $\ket{\omeganb}$, the ones with positive sign adding a positive charge $2e$ on the respective superconducting layer (thus removing a Copper pair), 
while the ones with negative sign adding a negative charge $-2e$ (thus creating a Cooper pair). 
Thanks to the commutators in (\ref{L-algebra}),
the number of such excitations in the two layers is counted by the conjugate momenta
$p_L^N$ and $p_R^N$, so that multiplying them by the Cooper pair charge yields the charge operators on the electrodes.
The charge excess between the two layers can then be expressed by their semidifference.
Recalling that a capacitive energy is quadratic in the excess charge,
the capacitive Hamiltonian term in the GNS representation can be conveniently written as
\begin{equation}
H^N_C = \mathcal{E}_C \left(\frac{p^N_L-p^N_R}{2}-n_g\right)^2\ ,
\label{H-capacity}
\end{equation}
where $\mathcal{E}_C$ represents the charging energy, while the constant parameter $n_g$ gives the excess charge 
induced by an external gate voltage as in (\ref{charge-qubit}).

At finite $N$, the Hamiltonian describing the charge qubit circuit can then be expressed as:
\begin{equation}
\label{charge-qubit-hamiltonian}
\mathfrak{H}^N \equiv H^N_\mathrm{free} + H^N_\mathrm{int} + H^N_C \ ,
\end{equation}
so that the corresponding time evolution operator is given by
\begin{equation}
\label{evolution}
U^N(t) = e^{-it\, \mathfrak{H}^N}.
\end{equation}
As already anticipated by Theorem~\ref{ch5.1-th:dyn} for a single superconductor,
only the last two terms in (\ref{charge-qubit-hamiltonian}) induce a nontrivial mesoscopic dynamics.
Indeed, we have

\begin{restatable}[Charge qubit dynamics]{theorem}{dynamicsqubit}
\label{ch5.2-th:dyn}
The microscopic time-evolution operator $U^N(t)$ has a well-defined mesoscopic limit,
\begin{equation}
\label{mesoscopic-dynamics}
\mlim_{N\to\infty} U^N(t) = U(t) \equiv e^{-it\, \mathfrak{H}}\ ,
\end{equation}
where the Hamiltonian generating the mesoscopic dynamics on the circle is given by
\begin{equation}
\label{mesoscopic-hamiltonian-1}
\mathfrak{H} = \mathcal{E}_C\left(\frac{\hat{p}_{\varphi_L}-\hat{p}_{\varphi_R}}{2}-n_g\right)^2
+2\lambda \mathfrak{c}_L \mathfrak{c}_R\, \cos\big(\hat{\varphi}_L-\hat{\varphi}_R\big) \ .
\end{equation}
\end{restatable}
\noindent
The proof is provided  in Appendix B.

Introducing the relative coordinate $\hat{\varphi}$ and momentum $\hat{p}_\varphi$ operators,
\begin{equation}\label{ch5.2-eq:relative-coords-def}
e^{i\hat{\varphi}} \equiv e^{i\hat{\varphi}_L}e^{-i\hat{\varphi}_R}\ ,\qquad  \hat{p}_\varphi \equiv \frac{\hat{p}_{\varphi_L}-\hat{p}_{\varphi_R}}{2}\ , 
\qquad [\hat{p}_\varphi,e^{i\hat{\varphi}}] = e^{i\hat{\varphi}}\ ,
\end{equation}
the mesoscopic Hamiltonian can be rewritten as
\begin{equation}
\label{mesoscopic-hamiltonian}
\mathfrak{H} = \mathcal{E}_C\left(\hat{p}_\varphi-n_g\right)^2
+\mathcal{E}_J\, \cos\hat{\varphi} \ ,\qquad \mathcal{E}_J=2\lambda \mathfrak{c}_L \mathfrak{c}_R\ ,
\end{equation}
thus retrieving the formal, phenomenological charge qubit Hamiltonian in (\ref{charge-qubit}). 
However, it should be stressed that $\mathfrak{H}$ in (\ref{mesoscopic-hamiltonian}) is a fully quantum operator
obtained as mesoscopic limit of the microscopic Hamiltonian (\ref{charge-qubit-hamiltonian}):
in contrast to the standard phenomenological
procedure, $e^{\pm i\hat\varphi}$ appearing in (\ref{mesoscopic-hamiltonian}) are Weyl-like operators,
not fixed-phase exponentials.

The Josephson current operator can now be retrieved by computing the time variation of the excess charge
in the junction in units of the Cooper pair charge:
\begin{equation}
\label{current}
J(t)= -\dot{\hat{p}}_\varphi= -i[\mathfrak{H}, \, \hat{p}_\varphi]= \mathcal{E}_J\, \sin\hat\varphi\ .
\end{equation}
Recalling the result (\ref{Thirring-2}), one sees that the critical Josephson current in the strong coupling regime is proportional
to $\mathcal{E}_J$ and therefore to the product of the modulus of the gaps relative to the $L$ and $R$ layers:
\begin{equation}
\label{critical-current}
\mathcal{E}_J = 2\lambda\, \Delta_L\, \Delta_R\ ,
\end{equation}
a result obtained in \cite{Aslamazov} using phenomenological methods (see also \cite{Barone,Bakurskiy}).

\begin{remark}
The above dependence of the critical current differs from the one found in~\cite{Ambegaokar}.
However, the approach there pursued is inapplicable
in the strong-coupling regime, as some finite contributions diverge when the energy levels 
of the electrons close to the Fermi surface are independent from their momenta, $\epsilon_{\vec{k}}=\epsilon$.
Nevertheless, assuming for simplicity identical superconductor layers,
from (\ref{critical-current}) one can recover the result
of \cite{Ambegaokar} in the regime of very small energies, $\varepsilon\ll2T_c\Delta$. 
Indeed, in this case, the consistency condition (\ref{consistency}) reduces to 
\begin{equation}
\label{consistency2}
2\Delta =\tanh(2\beta\Delta/\beta_c)\ ,
\end{equation}
so that, in this regime, the critical current can be rewritten as:
\begin{equation}
\mathcal{E}_J = \frac{\lambda \beta_c}{2}\, \bold{\Delta}  \tanh\left(\frac{\beta\bold{\Delta}}{2}\right)\ ,
\end{equation}
in terms of a suitably rescaled gap $\bold{\Delta}= 4T_c\, \Delta$.
Notice that it is precisely $\bold{\Delta}$ that should be identified with the measured gap,
as it has the dimension of an energy and reproduces the correct phenomenological behaviour~\cite{Rickayzen}-\cite{Abrikosov}.
Indeed, for small temperatures, $\beta\to\infty$, the r.h.s. of (\ref{consistency2}) approaches one,
so that $\bold{\Delta}(0)\simeq 2 T_c$. On the other hand, close to the critical temperature, $T\simeq T_c$,
by expanding the hyperbolic tangent, one instead finds: 
\hbox{$\bold{\Delta}(T)\simeq\sqrt{3}\, \bold{\Delta}(0)\,\big(1-T/T_c\big)^{1/2}$}, in good agreement with
the known results.
\end{remark} 

\begin{remark}
It should be stressed that our treatment of the charge qubit circuit, and in particular of the Josephson junction
in it, is fully gauge invariant, no phase has been fixed in the choice of collective fluctuation operators,
nor in the derivation of their mesoscopic limit and dynamics. This is in contrast with the usually adopted approaches,
based on the mean-field approximation, in which typically the relative phase of the the two-layer gaps
is fixed ({\it e.g.} see \cite{Unnerstall,Lauwers}).
In particular, the Josephson current operator in (\ref{current}) can have a non-vanishing expectation value
only on states with a fixed phase. Indeed, using the relative coordinates, 
any state on the circle can be written as:
\begin{equation}
|\psi\rangle = \frac{1}{2} \int_0^{2\pi} {\rm d}\varphi\ \psi(\varphi)\, |\varphi\rangle\ ,
\end{equation}
where $\psi(\varphi)$ is a suitable weight function, while the state $|\varphi\rangle\equiv \sum_n e^{in\varphi}|n\rangle$, 
with $|n\rangle$ as in (\ref{state-n}), is formally an eigenstate of the angle operator.
Clearly, the average of the Josephson operator (\ref{current}) on $|\psi\rangle$,
\begin{equation}
\big\langle \, J\, \big\rangle_\psi = \frac{\mathcal{E}_J}{2\pi} \int_0^{2\pi} {\rm d}\varphi\ |\psi(\varphi)|^2\ \sin \varphi\ ,
\end{equation}
will reproduce the expected value only for a smearing function $\psi(\varphi)$ peaked at a given phase
$\bar\varphi$, for which
\begin{equation}
\big\langle \, J\, \big\rangle_\psi = \mathcal{E}_J \ \sin \bar\varphi\ ,
\end{equation}
thus recovering the standard phenomenological result.
\end{remark}

%% file: sections/discussion.tex
\section{Discussion}

In dealing with many-body quantum systems, made of $N$ microscopic components, 
the relevant observables are collective ones, consisting of suitably scaled sums of microscopic operators. 
Among them, macroscopic averages that scale as the inverse of $N$ provide, in the large-$N$ limit, 
a description of the emerging commutative, henceforth classical, collective features of such quantum systems.
However, other relevant classes of collective observables can be constructed, the so-called quantum fluctuations,
scaling with different powers of $1/N$, while retaining quantum features in the large-$N$ limit;
for instance, whenever considering states with low correlation content, fluctuations behave as bosonic operators thus obeying canonical commutation relations.
These collective observables describe many-body physics at a mesoscopic scale,
in between the purely quantum behaviour of microscopic observables and the purely 
classical one of commuting macroscopic ones.

We have shown that quantum fluctuations are the most appropriate choice of observables for describing
the quantum behaviour of superconducting circuits based on Josephson junctions. The dynamics of
these systems involves collective phenomena that can not be described by looking at the behaviour
of finite-$N$ microscopic constituents, nor at a macroscopic, classical scale.
Instead, we have looked at the collective behaviour of superconducting junctions by means of two fluctuation quantities: 
the excess number of Cooper pairs on the junction and the phase difference between the superconducting condensates.
We have found that, in the large-$N$ limit, these quantities behave as conjugated canonical quantum operators, 
acting like momentum and angle variable for a particle on a circle inside an anharmonic potential.
Remarkably, their emergent mesoscopic dynamics
is generated by a quantum Hamiltonian which reproduces the phenomenological circuit equations,
without the need of a re-quantization of a semiclassical inspired model. As a byproduct, our fully quantum approach
provides a derivation of the temperature dependence of the critical current that, in the strong-coupling scenario, is not accessible by means of standard approximation methods.

Although, for sake of definiteness, our considerations have been focused on a particular class of superconducting circuits, 
those known as charge qubits, the presented techniques, based on the so-called GNS representation
and a strong-coupling, quasi-spin approach, are quite general, and can be applied to model
more complex circuits, {\it e.g.} involving flux-qubits and transmons. In this respect,
the possibility of the presence of more superconductive devices in the same circuit can allow the study
of their emerging mutual entanglement in a purely quantum mechanical setting, without any 
assumption based on phenomenological or semiclassical considerations;
this perspective is one of the most intriguing outcomes of our investigation.
\vskip .5cm

\noindent
\textbf{Acknowledgements}

FB acknolewdges financial support from PNRR MUR project PE0000023-NQSTI.

FC acknowledges funding from the Deutsche Forschungsgemeinschaft (DFG, German Research Foundation) under Project No. 435696605 and through the Research Unit FOR 5413/1, Grant No. 465199066 as well as from the European Union's Horizon Europe research and innovation program under Grant Agreement No. 101046968 (BRISQ).  FC~is indebted to the Baden-W\"urttemberg Stiftung for the financial support by the Eliteprogramme for Postdocs.

%% file: sections/appendix_single_layer.tex
\section{Single superconductor}

\label{app1}

In this Appendix we provide the proofs of Theorems~\ref{ch5.1-th:reconstruction}
and~\ref{ch5.1-th:dyn}: we begin with some useful estimates.

\subsection{Preliminary results}

\begin{proposition}
Given the definition of the collective spin operators in \eqref{collective-spin}, 
one finds the following norm estimate:
\begin{subequations}\label{app0-eq:norm-sxyz}
\begin{gather}
    \bignorm{S^N_{x,y,z}} = \frac{N}{2}\ ,
    \label{app0-eq:norm-sxyz-bis}
    \\
    \bignorm{S^N_\pm} \le \frac{N+1}{2}\ ,
    \label{app0-eq:norm-S^N_pm}
    \\
    \Bignorm{\Bigl[(S^N_+)^n,(S^N_-)^m\Bigr]} = \bigo{(N^{n+m-1})}\ .
    \label{app0-eq:norm-comm-macro-spin}
\end{gather}
\end{subequations}
\end{proposition}

\begin{proof}
The first result follows from the fact that the norm of an operator $X$ on a finite-dimensional 
Hilbert space amounts to the square root of the maximum eigenvalue of $X\dag X$, while
the second one by writing:
$$
S^N_\pm S^N_\mp = (\vec{S}^N)^2 - (S^N_z)^2 \pm S^N_z\ .
$$ 
The third bound can be proved by induction: for $n=1$ and $m=1$, 
\begin{equation}
    \Bignorm{\Bigl[S^N_+,S^N_-\Bigr]} = \norm{2S^N_z} = N = \bigo{(N)},
\end{equation}
while for $n=1$ and generic $m$, assuming~\eqref{app0-eq:norm-comm-macro-spin} to hold for $m-1$, 
\begin{align*}
&
\Bignorm{\Bigl[S^N_+,(S^N_-)^m\Bigr]}
\le \Bignorm{S^N_-\Bigl[S^N_+,(S^N_-)^{m-1}\Bigr]} + \Bignorm{\Bigl[S^N_+,S^N_-\Bigr](S^N_-)^{m-1}} \\
        &\hskip 1cm\le \Bignorm{S^N_-}\,\Bignorm{\Bigl[S^N_+,(S^N_-)^{m-1}\Bigr]}+ \Bignorm{2S^N_z}\,\Bignorm{S^N_-}^{m-1} \\
        & \hskip 2cm\le \frac{N+1}{2}\,\Bignorm{\Bigl[S^N_+,(S^N_-)^{m-1}\Bigr]} + N\,\left(\frac{N+1}{2}\right)^{m-1}= \bigo{(N^m)}.
\end{align*}
For generic $n$, we proceed by induction in a similar way.
\end{proof}

\begin{remark}
The previous estimates yields the following bounds for the norms of the phase operators in~\eqref{phase},
with $N\ge1$ and $m,n\in\N$:
\begin{subequations}
\label{app0-eq:norm-E_N}
\begin{gather}
    \Bignorm{(E_\pm^N)^m} \le \left(\frac{N+1}{2N \mathfrak{c}}\right)^m \le \left(\frac{1}{\mathfrak{c}}\right)^m = \bigo{(N^0)}    
    \label{app0-eq:norm-E_N-explicit}\ ,\\
    \Bignorm{\Bigl[(E_-^N)^n,(E_+^N)^m\Bigr]}=\bigo{(N^{-1})}.
    \label{app0-eq:comm-E^N}
\end{gather}
\end{subequations}
\end{remark}
\noindent
Concerning the correlation functions in~\eqref{F-correlation}, one has the following useful results:
\begin{corollary}\label{app0-corollary:useful-limits}
The following limit holds:
\begin{equation}\label{app0-eq:phase-op-exchange}
\lim_{N\to\infty}\Bigexpvalue{\Bigl(\prod_{j=1}^r (E_-^N)^{n_j}(E_+^N)^{m_j}\Bigr)}_\beta^N 
= \lim_{N\to\infty} \bigexpvalue{(E_-^N)^n(E_+^N)^m}_\beta^N ,
\end{equation}
for all choices of non negative integers $\{n_j,m_j\}_{j=1}^r$, where $n\equiv\sum_{j=1}^rn_j$ and $m\equiv\sum_{j=1}^rm_j$.
Furthermore, if a sequence of (local) observables $\{X^N\}$ satisfies
\begin{equation}
    \lim_{N\to\infty}\bigexpvalue{(X^N)\dag X^N}_\beta^N = \lim_{N\to\infty}\bignorm{X^N\ket{\omeganb}}^2 = 0,
\end{equation}
then,
\begin{gather}\label{app0-eq:useful-limit}
    \lim_{N\to\infty}\bigexpvalue{(E_-^N)^n(E_+^N)^m\, X^N}_\beta^N = 0.  
\end{gather}
\end{corollary}
\begin{proof}
Equality~\eqref{app0-eq:phase-op-exchange} means that in the large-$N$ limit the operators $E_\pm^N$ 
that alternate in the product at the left hand side of  the equality can be 
harmlessly regrouped. Indeed, the exchange of, say, $(E_-^N)^{n_1}$ and $(E_+^N)^{m_1}$ 
must be compensated by their commutator. However,
due to \eqref{app0-eq:norm-E_N-explicit} and \eqref{app0-eq:comm-E^N}, the modulus of  
\begin{equation}
    \Bigexpvalue{\left[(E_-^N)^{n_1},(E_+^N)^{m_1}\right]\,\Bigl(\prod_{j=2}^r (E_-^N)^{n_j}(E_+^N)^{m_j}\Bigr)}_\beta^N ,
\end{equation}
is bounded from above by 
\begin{equation}
\nonumber
    \Bignorm{\Bigl[(E_-^N)^{n_1},(E_+^N)^{m_1}\Bigr]}\,\Bigl(\prod_{j=2}^r \bignorm{(E_-^N)^{n_j}}\bignorm{(E_+^N)^{m_j}}\Bigr)
    \bignorm{\ket{\omeganb}}^2= \bigo{(N^{-1})}.
\end{equation}
The limit~\eqref{app0-eq:useful-limit} follows instead from the  
Cauchy-Schwarz inequality and the bound~\eqref{app0-eq:norm-E_N-explicit}:
\begin{align*}
        \Bigabs{\bigexpvalue{(E_-^N)^n(E_+^N)^m\, X^N}_\beta^N }
        & \le \Bignorm{(E_-^N)^n(E_+^N)^m}\, \bignorm{\ket{\omeganb}}\,\bignorm{X^N\ket{\omeganb}} \\
        & \le \Bigl(\frac{1}{\mathfrak{c}}\Bigr)^{n+m}\bignorm{X^N\ket{\omeganb}},
\end{align*}
where we also used that $\ket{\Omega^N_\beta}$ is normalized.
\end{proof}

\subsection{Theorem~\ref{ch5.1-th:reconstruction}: fluctuations in the large-$N$ limit}

The proof of Theorem~\ref{ch5.1-th:reconstruction} relies on the large-$N$ behaviour of the correlation functions in the strong-coupling BCS model, and their equivalence to the correlation functions evaluated in mean-field theory, and on  Corollary~\ref{app0-corollary:useful-limits}.

\reconstruction*

\begin{proof}

By means of the second algebraic relation in~\eqref{algebra-N}, 
we can exchange the operators inside the correlation functions on the left hand side 
of~\eqref{mesoscopic-correlations} above
and bring all unitaries generated by $p^N$ to the right. 
For example, consider the first term of the product:
\begin{equation}
    \begin{split}
        e^{i\alpha_1 p^N} (E_-^N)^{n_1} (E_+^N)^{m_1}
        & = e^{-in_1\alpha_1} (E_-^N)^{n_1} e^{i\alpha_1 p^N} (E_+^N)^{m_1} \\
        & = e^{-i(n_1-m_1)\alpha_1} (E_-^N)^{n_1} (E_+^N)^{m_1} e^{i\alpha_1 p^N}.
    \end{split}
\end{equation}
By iteration one finds:
\begin{align}
\nonumber
& \prod_{j=1}^r e^{i\alpha_j p^N} (E_-^N)^{n_j}(E_+^N)^{m_j} 
= \exp{\Bigl(i\sum_{j=1}^r\sum_{k=1}^j\alpha_{k}(m_j-n_j)\Bigr)}\\
 \label{app1-th:rec-eq:phase-pullout}
       &\hskip 5cm
        \times \Bigl(\prod_{j=1}^r (E_-^N)^{n_j}(E_+^N)^{m_j} \Bigr)e^{i\sum_{j=1}^r\alpha_j p^N},
\end{align}
so the phase factor in the first equality of~\eqref{mesoscopic-correlations} is retrieved. 
Moreover, due to~\eqref{on-vacuum}, 
the unitaries generated by $p^{N}$ leave $\ket{\omeganb}$ invariant. Then, we just need to prove that 
\begin{equation}
\label{app1-ch5.1-th:rec-eq:reduced-lim}
    \lim_{N\to\infty} \Bigexpvalue{ \prod_{j=1}^r (E_-^N)^{n_j}(E_+^N)^{m_j}}_\beta^N = \delta_{m,n}\ . 
\end{equation}
Recalling \eqref{app0-eq:phase-op-exchange}, it is sufficient to prove that:
\begin{equation}
\label{reduced-lim-2}
\lim_{N\to\infty}\bigexpvalue{(E_-^N)^{n}(E_+^N)^{m}}_\beta^N = \delta_{m,n}\ .
\end{equation}
Consider $n\ne m$ with $N>n,m>0$; recalling the definition(\ref{phase}),
\begin{equation*}
\label{app1-eq:delta-origin}
(E_-^N)^n(E_+^N)^m \ket{\omeganb}
    \propto\sum_{s,s_z}\sqrt{\rho^N(s,s_z)} \Gamma(s,s_z,m,n) \ \bigoplus_{\alpha=1}^{d(s)}\ket{s,s_z+m-n}_\alpha\otimes\ket{s,s_z}_\alpha\ ,
\end{equation*}
where $\ket{s,s_z+m-n}_\alpha=0$ if either $s_z+m-n>s$ or $s_z+m-n<-s$. Moreover, the real factor $\Gamma(s,s_z,m,n)$ 
comes from the action of the operators $E^N_\pm$ on the spin eigenstates:  thanks to~\eqref{app0-eq:norm-E_N-explicit}, it 
remains bounded as $N$ becomes large.  
Taking the scalar product of the previous vector 
with $\ket{\omeganb}$ yields
\begin{eqnarray*}
 &&
 \sum_{s',s_z'}\sum_{s,s_z}\sum_{\alpha_1=1}^{d(s)}\sum_{\alpha_2=1}^{d(s')}\,\sqrt{\rho^N(s',s'_z)\, \rho^N(s,s_z)}\Gamma(s,s_z,m,n)\\
 &&\hskip 3cm 
\times\  _{\alpha_1}\!\braket{s',s_z' | s,s_z+m-n}_{\alpha_2}\ _{\alpha_1}\!\braket{s',s_z' | s,s_z}_{\alpha_2} = 0.
\end{eqnarray*}
Let $m=n$, then
\begin{equation}\label{app1-th:rec-eq:boh}
    \bigexpvalue{ (E_-^N)^n(E_+^N)^n }_\beta^N
    = \omeganb\biggl( \Bigl(\frac{1}{\mathfrak{c}N}\Bigr)^{2n} (S^N_-)^{n}(S^N_+)^{n} \biggr).
\end{equation}
Thanks to~\eqref{Thirring}, the large-$N$ limit of the right hand side of~\eqref{app1-th:rec-eq:boh} reads:
\begin{equation}\label{app1-th:rec-eq:lim-mf}
    \lim_{N\to\infty}\int_0^{2\pi} \frac{d\varphi}{2\pi}\, \Bigl(\frac{1}{\mathfrak{c}N}\Bigr)^{2n}
    \langle\!\langle  (S^N_-)^{n}(S^N_+)^{n} \rangle\!\rangle_\beta^N\ .
\end{equation}
Recall that $\langle\!\langle  \,\cdot\, \rangle\!\rangle_\beta^N$ was defined in (\ref{mf-average}) 
as the Gibbs state corresponding to the mean-field version of the 
strong-coupling BCS Hamiltonian. Explicitly, one has
\begin{equation*}
       \frac{(S^N_-)^{n}(S^N_+)^{n}}{N^{2n}}
        = \sum_{q_1, \cdots, q_n=1}^N\sum_{p_1,\cdots,p_n=1}^N\frac{\sigma^{(q_1)}_-}{N}\cdots\frac{\sigma^{(q_{n})}_-}{N}\frac{\sigma^{(p_1)}_+}{N}\cdots\frac{\sigma^{(p_{n})}_+}{N}\ ,
\end{equation*}
where all upper indices of type $q$, respectively $p$ must differ, for $\sigma_\pm^2=0$. Furthermore, when $q_i=p_j$ the corresponding sum $\sum_{q_i}$
vanishes as $1/N$. 
Then, the  mean-field expectations carry no correlations between spins at different locations 
and are site independent. Therefore, 
\begin{align*}
    &\langle\!\langle \bigl(\sigma^{(q_1)}_-\cdots\sigma^{(q_{n})}_-\bigr)\bigl(\sigma^{(p_1)}_+\cdots\sigma^{(p_{n})}_+\bigr) \rangle\!\rangle_\beta^N= \\
    &\hskip 3cm= \bigl(\langle\!\langle\sigma_-\rangle\!\rangle_\beta^N\bigr)^n
    \bigl(\langle\!\langle\sigma_+\rangle\!\rangle_\beta^N\bigr)^n = \Delta^{2n}=\mathfrak{c}^{2n},
\end{align*}
where \eqref{Thirring-2} has been used. Thus, the limit in~\eqref{app1-th:rec-eq:lim-mf} reduces to
\begin{equation*}
    \lim_{N\to\infty}\int_0^{2\pi} \frac{d\varphi}{2\pi}\,\Bigl(\frac{1}{\mathfrak{c}N}\Bigr)^{2n} \frac{N!}{(N-2n)!}
   \, \mathfrak{c}^{2n} = 1,
\end{equation*}
and \eqref{reduced-lim-2} is proven.

\end{proof}

\subsection{Theorem~\ref{ch5.1-th:dyn}: fluctuations dynamics for a single superconductor}

In the following we prove Theorem~\ref{ch5.1-th:dyn} for the dynamics of the single BCS system. We shall deal with the evolution operator
\begin{equation}
\mathcal{U}^N(t)\equiv e^{-it(\mathcal{H}^N\otimes\identity-\identity\otimes \mathcal{H}^N)} =
e^{-it({H}^N\otimes\identity-\identity\otimes {H}^N)}\, e^{-it\, (2\mu  p^N)}\ .
\label{time-evolution-2}
\end{equation}
We show that the Hamiltonian $H^N$ in (\ref{BCS-spin-2}) does not contribute to the time evolution of the fluctuation operators in the large-$N$ limit. In order to do so, we introduce and characterize a new operator.

\begin{proposition}
\label{app1-prop:help-dynamics}
The operator $E^N(t)$ defined by
\begin{equation}\label{app1-ch5.1-th:dyn-eq:E^N_t-def}
    E_\pm^N(t)\equiv e^{-it({H}^N\otimes\identity-\identity\otimes {H}^N)}\,
    E_\pm^N\, e^{it({H}^N\otimes\identity-\identity\otimes {H}^N)}\ ,
\end{equation}
is such that
\begin{equation}\label{app1-eq:E^N_t-explicit-expr}
    E^N_\pm(t) = E_\pm^N\, W^N(t),
\end{equation}
where $W^N(t)$ is the following unitary time-evolution generated by a Hamiltonian $K^N$:
\begin{equation}\label{K^N-def}
    W^N(t) \equiv e^{-itK^N},\quad K^N \equiv -2\epsilon\, \identity +  \frac{2 T_c}{N}S^N_z\ .
\end{equation}
Furthermore, $W^N(t)$ satisfies the following properties for $m\in\N$:
\begin{subequations}\label{app1-eq:W^N(t)-properties}
\begin{gather}
    \Bignorm{\Bigl[E_\pm^N,W^N(t)\Bigr]} = \bigo{(N^{-1})}\ ,
    \label{app1-eq:comm-E^N-W^N(t)-vanish}\\ 
    \lim_{N\to\infty} \bigexpvalue{(W^N(t))^m}^N_\beta = 1\ ,
    \label{app1-eq:W^N(t)-vanish-exp-value}\\ 
    \lim_{N\to\infty} \bignorm{\bigl((W^N(t))^m-1\bigr)\ket{\omeganb}} = 0\ . 
    \label{app1-eq:W^N(t)-vanish-norm}
\end{gather}
\end{subequations}
\end{proposition}
\begin{proof}
Since $E^N_-=(E^N_+)^\dagger$, we concentrate on $E^N_+$. The expression~\eqref{app1-eq:E^N_t-explicit-expr} for $E^N_+(t)$ can be computed by exploiting the Baker-Campbell-Hausdorff formula:
\begin{equation}\label{app1-dim:flim-dyn-eq:BCH}
 E^N_+(t) = E_+^N\,-\,i\,t\,\big[(H^N\otimes\identity-\identity\otimes H^N),\,E_+^N\big]+  \mathcal{O}(t)\ .
\end{equation}
Since $E_+^N$ is proportional to $S^N_+\otimes\identity$, all commutators with $\identity\otimes H^N$ vanish, 
while
\begin{align}
\nonumber
\Bigl[H^N, E_+^N\Bigr] 
    &= \frac{1}{\mathfrak{c}N} [H^N,S^N_+] 
    = \frac{1}{\mathfrak{c}N} \Biggl(-2\epsilon S^N_+ + \frac{2T_c}{N}S^N_+S^N_z\Biggr)\\
    \label{app1-dim:flim-dyn-eq:commut-H^N-E^N}
    &= E_+^N\, K^N\ .
\end{align}
Since $\Bigl[H^N\,,\,K^N\Bigr] = 0$,  substituting~\eqref{app1-dim:flim-dyn-eq:commut-H^N-E^N} into~\eqref{app1-dim:flim-dyn-eq:BCH} yields
\begin{equation*}
        \Bigl[H^N, \Bigl[H^N, E_+^N\Bigr]\Bigr] 
        = \Bigl[H^N, E_+^N\Bigr]\, K^N 
        = E_+^N\, (K^N)^2 .
\end{equation*}
Iterating the procedure,~\eqref{app1-dim:flim-dyn-eq:BCH} becomes
\begin{equation*}\label{app1_dim:flim-dyn_eq:time-ev-E^N-explicit-expr}
        E^N_+(t) = E_+^N\, \sum_{k=1}^{\infty}\frac{(-it)^k}{k!}(K^N)^k  
        = E_+^N\, e^{-itK^N}
        = E_+^N\, W^N(t),
\end{equation*}
thus proving the first statement in~\eqref{app1-eq:E^N_t-explicit-expr}. On the other hand, from
$$
        \Bigl[E_+^N\,,\,W^N(t)\Bigr]= E_+^N\, W^N(t) \Bigl(1-e^{-it2 T_c/N}\Bigr)\ ,
$$
and~\eqref{app0-eq:norm-E_N}, one estimates
$$
    \Bignorm{\Bigl[E_+^N\,,\,W^N(t)\Bigr]} \le \underbrace{\norm{E_+^N}}_{{}=1/\mathfrak{c}}\norm{W^N(t)} \Bigabs{1-e^{-it2T_c/N}} = \bigo{(N^{-1})},
$$
so that~\eqref{app1-eq:comm-E^N-W^N(t)-vanish} is proven. Let us now turn our attention to~\eqref{app1-eq:W^N(t)-vanish-exp-value}. We have
\begin{equation}
\label{app1-eq:expvalue-wnt}
    \bigexpvalue{(W^N(t))^m}_\beta^N =e^{it2m\epsilon}\omeganb\bigl(e^{-it2m T_c S^N_z/N}\bigr).
\end{equation}
Thanks to the results of Thirring in~\cite{Thirring1}, we know that $S^N_z/N$ converges to $\epsilon/T_c$, hence
\begin{equation}\label{app1-eq:proof-W^N(t)-vanish}
    \lim_{N\to\infty}\bigexpvalue{(W^N(t))^m}_\beta^N = e^{itm2\epsilon}e^{-itm 2\epsilon} = 1\ .
\end{equation}
The limit in~\eqref{app1-eq:W^N(t)-vanish-norm} follows from
\begin{equation}
    \begin{split}
        \bignorm{\bigl((W^N(t))^m-1\bigr)\ket{\omeganb}}^2 
        & = \bigexpvalue{\bigl((W^N(t))^m-1\bigr)\dag \bigl((W^N(t))^m-1\bigr)}_\beta^N \\
        & = 2\bigl(1-\Re{\expvalue{(W^N(t))^m}_\beta^N}\bigr),
    \end{split}
\end{equation}
and the last term goes to zero thanks to~\eqref{app1-eq:proof-W^N(t)-vanish}.
\end{proof}

Proposition~\ref{app1-prop:help-dynamics} together with Corollary~\ref{app0-corollary:useful-limits} can be used to prove the following Theorem  
which characterises the evolution operator in the mesoscopic limit. 

\dynamics*

\noindent
Notice that we are going to show that the  strong coupling Hamiltonian does not induce any mesoscopic dynamics, in agreement with the fact that the 
thermal state is an equilibrium state with respect to this evolution.

Recalling the definition of the mesoscopic limit in~\eqref{ch5.1-eq:def-fluc-lim}, we shall now introduce some useful short-hand notation 
to deal with the matrix elements appearing in the formulas. For any $n\in\Z$ we define the states
\begin{eqnarray}
&&
\ket{n} \equiv e^{in\hat\varphi}\ket{0}\ , \qquad
\ket{n}_\beta^N \equiv 
\begin{cases}
(E_{+}^{N})^{n}\ket{\omeganb}\,&\textup{if $n\ge0$,} \\
(E_{-}^{N})^{-n}\ket{\omeganb}\,&\textup{if $n<0$.} \\
\end{cases}
\end{eqnarray}
The states $\ket{n}$ form the orthonormal basis~\eqref{state-n} in the Hilbert space of the limit Heisenberg algebra on the circle; the states $\ket{n}_\beta^N$ 
are instead quasi-spin states. Notice that the integer labels of the states can be either positive or negative, while the powers of the collective quasi-spin operators operators are always non-negative.
In this way for example, the main statement of Theorem~\eqref{ch5.1-th:dyn} can be recast as
\begin{equation}
\label{eq:mlim-single-superc-evolution-matrixelements}
\lim_{N\to\infty} \braket{ n | \mathcal{U}^N(t) | m }_\beta^N = \braket{ n | \mathcal{U}(t) | m },
\end{equation}
for any $n,m\in\Z$. We can now turn our attention to the Theorem itself.

\begin{proof}
Though we need prove~\eqref{eq:mlim-single-superc-evolution-matrixelements} for any $n,m\in\Z$, we restrict to $n,\,m\geq 0$ since the proofs for all other cases are analogous.

Recalling the algebraic relation~\eqref{algebra-N}, and defining 
$\widetilde{\mathcal{U}}^N(t) \equiv e^{-it({H}^N\otimes\identity-\identity\otimes {H}^N)}$, one has:
\begin{align*}\label{app1-th:dyn-eq:lhs-phase-pullout}
 \braket{ n | \mathcal{U}^N(t) | m }_\beta^N = e^{-it2\mu m} \bigexpvalue{(E_-^N)^n\, \widetilde{\mathcal{U}}^N(t)\,(E_+^N)^m\, e^{-it2\mu p^N}}_\beta^N = e^{-it2\mu m} \braket{ n | \widetilde{\mathcal{U}}^N(t) | m }_\beta^N ,
\end{align*}
where the last equality follows from $p^N\ket{\omeganb}=0$, as proven in~\eqref{on-vacuum}. 
Similarly, using the Heisenberg algebra on the circle, 
\begin{equation*}
\label{app1-th:dyn-eq:rhs-phase-pullout}
        \braket{n | \mathcal{U}(t) | m}  = 
        e^{-it2\mu m} \braket{n | m} = e^{-it2\mu m} \delta_{n,m}\ .
\end{equation*}
Then, it remains to prove that 
\begin{equation*}\label{app1-dim:flim-dyn-eq:allwemustprove}
    \lim_{N\to\infty}   \braket{ n | \widetilde{\mathcal{U}}^N(t) | m }_\beta^N
    =
    \delta_{n,m}
    = \lim_{N\to\infty}
        \bigexpvalue{(E_-^N)^n (E_+^N)^m}\ ,
\end{equation*}
where the second equality  follows from \eqref{reduced-lim-2}. 
We now have
\begin{align*}
\label{app1_dim:flim-dyn_eq:expvalue-trick}
\braket{ n | \widetilde{\mathcal{U}}^N(t) | m }_\beta^N
& = \bigexpvalue{(E_-^N)^n (E^N_+(t))^m\, \widetilde{\mathcal{U}}^N(t)}_\beta^N =  \bigexpvalue{(E_-^N)^n (E^N_+(t))^m}_\beta^N = \\
& = \bigexpvalue{(E_-^N)^n\bigl(E_+^N\,W^N(t))^m}_\beta^N\ .
\end{align*}
For the first equality, we used~\eqref{app1-ch5.1-th:dyn-eq:E^N_t-def} 
and the unitarity of $\widetilde{\mathcal{U}}^N(t)$; for the second one that
$\widetilde{\mathcal{U}}^N(t)\,\ket{\omeganb}=\ket{\omeganb}$.
Finally, we substituted $E^{N}_+(t)= E^{N}_+W^{N}_{t}$ from~\eqref{app1-eq:E^N_t-explicit-expr} in Proposition~\ref{app1-prop:help-dynamics} which asserts
that $\norm{[E^{N}_+,W^N(t)]}=\bigo{(N^{-1})}$. Therefore, in the large-$N$ limit we can exchange 
the position of all the $E_+^N$ and the $W^N(t)$, so that:
\begin{align*}
   	\lim_{N\to\infty}\braket{ n | \widetilde{\mathcal{U}}^N(t) | m }_\beta^N & = \lim_{N\to\infty}\bigexpvalue{(E_-^N)^n\,(E_+^N\,W^N(t))^m}_\beta^N \\
	& =	\lim_{N\to\infty}\bigexpvalue{(E_-^N)^n(E_+^N)^m\,(W^N(t))^m}_\beta^N\ .
\end{align*}
Then, it remains to be proved that
\begin{equation*}\label{app1-th:dyn-eq:final}
    \lim_{N\to\infty}\bigexpvalue{(E_-^N)^n(E_+^N)^m\,\bigl((W^N(t))^m-1\bigr)}_\beta^N = 0\ .
\end{equation*}
From~\eqref{app1-eq:W^N(t)-vanish-norm} in Proposition \ref{app1-prop:help-dynamics}, the norm $\norm{\bigl((W^N(t))^m-1\bigr)\ket{\omeganb}}$ vanishes when $N\to\infty$. Therefore, we can directly apply~\eqref{app0-eq:useful-limit} from Corollary~\ref{app0-corollary:useful-limits}, using $X^N = (W^N(t))^m-1$. It follows that the limit on the right hand side of the previous equation vanishes as well.
\end{proof}

Finally, we provide a result that will often be used in Appendix~\ref{app2}. 

\begin{proposition} 
\label{app1-prop:E^N_t-vanish}
Let $\{m_j,n_j\in\N,\, j=1,\dots,r\}$ be a finite sequence of non-negative integers. Then the following mesoscopic limit holds:
\begin{equation}\label{app1-eq:mlim-ENt-to-EN}
    \mlim_{N\to\infty}\biggl( \prod_{j=1}^r (E_-^N(t))^{n_j} (E_+^N(t))^{m_j}\biggr) = \mlim_{N\to\infty} (E_-^N)^{n} (E_+^N)^{m} = \delta_{n,m},
\end{equation}
where we set $m=\sum_{j=1}^rm_j$, $n=\sum_{j=1}^rn_j$.
\end{proposition}
\begin{proof}
Let us set
\begin{equation}
\begin{split}
\label{eq:Pt-def}
P^{N}_t(\{m\},\{n\}) & \equiv \Bigl(\prod_{j=1}^r (E_-^N(t))^{n_j} (E_+^N(t))^{m_j}\Bigr) = \Bigl(\prod_{j=1}^r\bigl( (W^N(t))\dag E_-^N \bigr)^{n_j} \bigl( E_+^N\, W^N(t) \bigr)^{m_j}\Bigr) = \\
	& = (E_-^N)^{n}(E_+^N)^{m}(W^N(t))^{m-n} + \Theta^N.
\end{split}
\end{equation}
In the second equality we directly substituted~\eqref{app1-eq:E^N_t-explicit-expr}; while in the third one we exchanged the operators $W^N(t)$ and $E^N_\pm$, thus obtaining a remainder $\Theta^N$ containing the necessary commutators. Then, the following bound holds
\begin{equation}
\begin{split}
\label{eq:norm-prod-ENt-bound}
\Bignorm{P^{N}_t(\{m\},\{n\}) - (E_-^N)^{n}(E_+^N)^{m}(W^N(t))^{m-n} } = \norm{\Theta^N} = \mathcal{O}(N^{-1}).
\end{split}
\end{equation}
Indeed, thanks to~\eqref{app0-eq:comm-E^N} and~\eqref{app1-eq:comm-E^N-W^N(t)-vanish}, all commutators vanish as $1/N$.

In order to prove the first equality of~\eqref{app1-eq:mlim-ENt-to-EN}, we consider the matrix element
\begin{equation}
\mathcal{M}^{(N)} \equiv \braket{n' | \left(P^{N}_t(\{m\},\{n\}) - (E_-^N)^{n} (E_+^N)^{m}\right) | m' }^N_\beta,
\end{equation} 
where in general $m',n'\in\Z$ and show that they vanish as $N$ grows large. It is sufficient to restrict to $m',\,n'\geq0$ as the proofs for all other cases are analogous.
Thanks to the bound~\eqref{eq:norm-prod-ENt-bound} we have
\begin{equation}
\label{eq:matrix-el}
\begin{split}
\abs{\mathcal{M}^{(N)}} \le	& \Bigabs{\braket{n' | (E_-^N)^{n} (E_+^N)^{m} ( (W^N(t))^{m-n} - 1) | m' }^N_\beta} + \mathcal{O}(N^{-1}) \\
						\le	& \Bigabs{\expvalue{(E_-^N)^{n+n'}(E_+^N)^{m+m'}( (W^N(t))^{m-n} - 1)}_\beta^N} + \mathcal{O}(N^{-1}),
\end{split}
\end{equation}
where in the second line we further exchanged $(W^N(t))^{m-n}$ and $(E_+^N)^{m'}$ at the cost of introducing a commutator, which nonetheless vanishes in the large-$N$ limit thanks once again to~\eqref{app1-eq:comm-E^N-W^N(t)-vanish}.

Finally, from~\eqref{app1-eq:W^N(t)-vanish-norm} in Proposition \ref{app1-prop:help-dynamics}, the norm $\norm{\bigl((W^N(t))^m-1\bigr)\ket{\omeganb}}$ vanishes when $N\to\infty$. Therefore, we can directly apply~\eqref{app0-eq:useful-limit} from Corollary~\ref{app0-corollary:useful-limits}, using $X^N = (W^N(t))^m-1$. It follows that the limit on the right hand side of~\eqref{eq:matrix-el} vanishes as well, thus proving the mesoscopic limit in~\eqref{app1-eq:mlim-ENt-to-EN}.

\end{proof}

%% file: sections/appendix_charge_qubits.tex
\section{Mesoscopic dynamics: charge qubits}
\label{app2}

In this Appendix we provide the proof of Theorem~\ref{ch5.2-th:dyn} for the dynamics of a charge qubit system. In Section~\ref{subsec:rel-coords-and-notation} we begin by introducing the relative coordinates for the charge qubit system and some new notations.

%
\subsection{Notation and relative coordinates}
\label{subsec:rel-coords-and-notation}

As in the previous appendix, we shall adopt the following notation for vectors in the Hilbert space on the circle, and for quasi-spin vectors at \hbox{finite $N$}; 
for any $n_L,\,n_R\in\Z$:
\begin{align}
&\ket{n_{L},n_{R}} \equiv e^{i(n_L\hat\varphi_L+n_R\hat\varphi_R)}\ket{0}\ , \\
&\ket{n_{L},n_{R}}_\beta^N \equiv 
\begin{cases}
(E_{L,+}^{N})^{n_L}(E_{R,+}^{N})^{n_R}\ket{\omeganb}\,&\textup{if $n_L\ge0$, $n_R\ge0$} \\
(E_{L,-}^{N})^{-n_L}(E_{R,+}^{N})^{n_R}\ket{\omeganb}\,&\textup{if $n_L<0$, $n_R\ge0$} \\
(E_{L,+}^{N})^{n_L}(E_{R,-}^{N})^{-n_R}\ket{\omeganb}\,&\textup{if $n_L\ge0$, $n_R<0$} \\
(E_{L,-}^{N})^{-n_L}(E_{R,-}^{N})^{-n_R}\ket{\omeganb}\,&\textup{if $n_L<0$, $n_R<0$} \ .\\
\end{cases}
\end{align}
The mesoscopic limit to be proved in Theorem~\ref{ch5.2-th:dyn} then reads 
\begin{equation}
\label{app2-eq:meso-lim-explicit-simplified-not}
    \lim_{N\to\infty} \braket{n_L',n_R' | U^N(t) |n_L,n_R}_{N,\beta} = \braket{n_L',n_R' | U(t) |n_L,n_R}\ , \quad \forall n_{L/R},n_{L/R}'\in\Z\ .
\end{equation}
Let us introduce the \emph{relative coordinates} for both the microscopic and mesoscopic system:
\begin{subequations}
\label{eq:relative-coords-def}
\begin{align}
& \hat p_\varphi \equiv \frac{\hat p_{\varphi_L}-\hat p_{\varphi_R}}{2}\ , \qquad e^{i\hat\varphi} \equiv e^{i\hat\varphi_L}e^{-i\hat\varphi_R}\ ,    \\
& p^N \equiv \frac{p_{\varphi_L^N}-p_{\varphi_R^N}}{2}\ , \qquad (\mathfrak{E}_\pm^N)^m \equiv (E_{\pm,L}^{N})^m(E_{\mp,R}^N)^m\ ,\quad m>0\ .
\end{align}
\end{subequations}
The usual algebraic relations for momentum-angle variables hold, that is, for $\alpha\in\R$ and $m\in\N$,
\begin{eqnarray}
\label{app2-eq:comm-rels-a}
&& [\hat{p}_\varphi,e^{\pm i\hat\varphi}] = \pm e^{i\hat\varphi}\ , \qquad 
e^{i\alpha \hat p_\varphi}e^{ im\hat\varphi} = e^{ im\alpha} e^{ im\hat\varphi}e^{i\alpha \hat p_\varphi}\ ,  \\
\label{app2-eq:comm-rels-b}
&& \Bigl[p^N\,,\mathfrak{E}_\pm^N\Bigr] = \pm \mathfrak{E}_\pm^N\ , \qquad 
e^{i\alpha p^N}(\mathfrak{E}_\pm^N)^{ m}   = e^{ \pm im\alpha} (\mathfrak{E}_\pm^N)^m e^{i\alpha p^N}\ .
\end{eqnarray}
Notice that the bound~\eqref{app0-eq:norm-E_N-explicit} extends directly to $\mathfrak{E}^N$:
\begin{equation}
\label{app2-eq:norm-E^N}
    \Bignorm{(\mathfrak{E}_\pm^N)^m}\le\left(\frac{1}{\mathfrak{c}_L \mathfrak{c}_R}\right)^m,\quad m>0\ .
\end{equation}
Recall the definitions of the capacitive term in~\eqref{H-capacity}, of the tunneling term in~\eqref{interaction}
and of the mesoscopic Hamiltonian~\eqref{mesoscopic-hamiltonian}:
\begin{eqnarray}
\label{H-pieces}
&& H^N_C = \mathcal{E}_\mathrm{C}(p^N-n_\mathrm{g})^2,\quad
    H^N_\mathrm{int} = \frac{\mathcal{E}_\mathrm{J}}{2}\bigl(\mathfrak{E}_+^N + \mathfrak{E}_-^N\bigr)\,,\\
\label{app2-eq:mesoscopic-hamiltonian}
&&
\label{circle-H}
H =\mathcal{E}_\mathrm{C}(\hat p_\varphi-n_\mathrm{g})^2+\mathcal{E}_\mathrm{J} \cos\hat \varphi\ .
\end{eqnarray}

Let us now split both the finite-$N$ Hamiltonian (\ref{charge-qubit-hamiltonian}) and the mesoscopic one in two pieces, a free Hamiltonian and a perturbing term:
\begin{equation}
   \mathfrak{H}^N = H^N_0+H^N_1,\quad H = H_0+H_1.
\end{equation}
where
\begin{subequations}
\begin{gather}
    H^N_0 =H^N_\mathrm{free}+H^N_\mathrm{C}\ ,\qquad H^N_1 = H^N_\mathrm{int}\ , \\
    H_0 = 4\mathcal{E}_\mathrm{C}(\hat p_\varphi-n_\mathrm{g})^2\ ,\qquad H_1 =\mathcal{E}_\mathrm{J}\cos\hat\varphi\ ,
\end{gather}
\end{subequations}
with corresponding  time-evolution operators:
\begin{equation}
    U^N_0(t) = e^{-itH^N_0}\ , \qquad U_0(t) = e^{-itH_0}\ .
\end{equation}
Notice that the free microscopic dynamics $U^{N}_{0}(t)$ generated by $H^N_0$ is made of two commuting contributions: one coming from the BCS Hamiltonian, and the other from the capacitive part, namely
\begin{equation}\label{app2-eq:micro-free-evolutor-factorization}
    U^N_0(t) = U^N_\mathrm{free}(t)\, U^N_\mathrm{C}(t)\ ,
\end{equation}
where
\begin{equation*}
    U^N_\mathrm{free}(t) = e^{-it\,H^N_\mathrm{free}}\ ,\qquad U^N_\mathrm{C}(t) = e^{-itH^N_\mathrm{C}}\ .
\end{equation*}

\subsection{Josephson Junctions}
Theorem~\ref{ch5.1-th:dyn} states that, in the case of a single supercondutor, 
the term corresponding to $U^N_\mathrm{free}(t)$, does not contribute 
to the mesoscopic dynamics. Clearly, the same result holds in presence of two independent superconductors:
\begin{equation}\label{app2-eq:bcs-vanish}
\lim_{N\to\infty}\braket{n_L',n_R' | U^N_\mathrm{free} |n_L,n_R} = \delta_{n_L,n_L'}\,\delta_{n_L,n_L'}\ ,
\end{equation}
for any $n_{L/R},\,n_{L/R}'\in\Z$. Analogously, extending the definition in~\eqref{app1-ch5.1-th:dyn-eq:E^N_t-def} to two superconductors,
\begin{equation}
\label{app2-eq:E^N_t-def}
    \mathfrak{E}^N_\pm(t) \equiv U^N_\mathrm{free}(t)\ \mathfrak{E}_\pm^N\ \bigl(U^N_\mathrm{free}(t)\bigr)\dag\ .
\end{equation}
Moreover, from Proposition~\ref{app1-prop:E^N_t-vanish} it follows that
\begin{equation}
\label{app2-eq:time-ev-phase-op-vanish}
\mlim_{N\to\infty}\biggl( \prod_{j=1}^r (\mathfrak{E}_-^N(t))^{n_j} (\mathfrak{E}_+^N(t))^{m_j}\biggr) = \mlim_{N\to\infty} (\mathfrak{E}_-^N)^{n} (\mathfrak{E}_+^N)^{m} = \delta_{n,m}.
\end{equation}
Notice also that $\mathfrak{E}^N_\pm(t)$ satisfy the usual algebraic relations, thanks to the commutativity of $H^N_\mathrm{free}$ and $p^N$:
\begin{equation}\label{app2-eq:comm-rels-time-dep}
    \Bigl[p^N,\, \mathfrak{E}^N_\pm(t)\Bigr] = \pm \mathfrak{E}^N_\pm(t)\ ,\qquad  e^{i\alpha p^N}(\mathfrak{E}_\pm^N(t))^{ m} = 
    e^{\pm im\alpha} (\mathfrak{E}_\pm^N(t))^{ m} e^{i\alpha p^N}\ ,\quad \alpha\in\R,\ m\in\N\ .
\end{equation}

The capacitive Hamiltonians $H_\mathrm{C}=\mathcal{E}_\mathrm{C}(\hat p_\varphi-n_\mathrm{g})^2 $ and $H^{N}_\mathrm{C}$ are instead functions only of the relative angular momenta $\hat p_\varphi$ and  $p^N$.
In order to compactify the notation, let
\begin{equation}\label{app2-eq:xi-def}
    \xi(X) \equiv \mathcal{E}_\mathrm{C}(X-n_\mathrm{g})^2\ ,
\end{equation}
for any operatorial or scalar quantity $X$. 
For instance, $H_\mathrm{C}=\xi(p_\varphi)$, $H^{N}_\mathrm{C}=\xi(p^{N}_\varphi)$, so that their eigenvalue equations read:
\begin{gather}\label{app2-eq:capacitor-spectrum}
    H_\mathrm{C}\ket{n_{L},n_{R}} = \xi(n)\ket{n_{L},n_{R}}\ , \\
    H_\mathrm{C}^N\ket{n_{L},n_{R}}_\beta^N = \xi(n)\ket{n_{L},n_{R}}_\beta^N\ , 
\end{gather}
with
\begin{equation}\label{app2-eq:n-def}
    n = \frac{n_L-n_R}{2}\ ,
\end{equation}
where one exploits~\eqref{eq:relative-coords-def}. Moreover, for operatorial and scalar quantities $X,Y$ we set
\begin{equation}\label{app2-eq:dxi-def}
    \Delta\xi(X,Y) \equiv \xi(X+Y)-\xi(X),
\end{equation}
so that 
\begin{equation}
    \Delta\xi(X, \pm 1) = \xi(X\pm1)-\xi(X) =  \mathcal{E}_\mathrm{C}\Big(1\pm2(X-n_\mathrm{g})\Big)\ ,
\end{equation}
is a linear function of $X$. Therefore, exploiting \eqref{app2-eq:comm-rels-a},  \eqref{app2-eq:comm-rels-b} 
and~\eqref{app2-eq:comm-rels-time-dep} we can write, for $m\in\N$ and $t\in\R$,
\begin{subequations}
\begin{gather}\label{app2-eq:dxi-eiphi-commutation}
    e^{-it\Delta\xi(\hat p_\varphi,\pm 1)} e^{im\hat\varphi} = e^{im\hat\varphi} e^{-it\Delta\xi(\hat p_\varphi+m,\pm 1)}\ ,    \\
    e^{-it\Delta\xi(p^N,\pm 1)} (\mathfrak{E}_+^N)^m = (\mathfrak{E}_+^N)^m\, e^{-it\Delta\xi(p^N\pm m,\pm 1)}\ , \\
    e^{-it\Delta\xi(p^N,\pm 1)} (\mathfrak{E}^N_+(t))^m = (\mathfrak{E}^N_+(t))^m e^{-it\Delta\xi(p^N\pm m,\pm 1)}\ .
\end{gather}
\end{subequations}
Analogous expressions hold true for the operators $\mathfrak{E}_-^N$ and $\mathfrak{E}^N_-(t)$.

\subsection{The Dyson series: useful results}\label{app2-sec:dyson-general}
Let us consider $H^N_1$, $H_1$ as perturbations with respect to the free Hamiltonians $H^N_0$ and $H_0$. Let us define the time-dependent potentials
\begin{subequations}
\begin{gather}
    V^N(t) \equiv U^N_0(t) H_1^N \bigl( U^N_0(t) \bigr)\dag\  , \label{app1-eq:V^N(t)-def} \\
    V(t) \equiv U_0(t) H_1 U_0\dag(t), \label{app1-def:time-dep-pot}
\end{gather}
\end{subequations}
and  introduce the finite $N$, quasi-spin and mesoscopic Dyson series
\begin{subequations}
\label{eqs:dyson-def}
\begin{align}
\mathcal{D}^N(t) \equiv 1 + \sum_{k=1}^{\infty}(-i)^{k}	\int_0^t dt_1\cdots \int_0^{t_{k-1}} dt_{k} V^N(t_{k})\cdots V^N(t_{1})\ ,  \label{DysonN} \\
\mathcal{D}(t) \equiv 1 + \sum_{k=1}^{\infty}(-i)^{k}	\int_0^t dt_1\cdots \int_0^{t_{k-1}} dt_{k} V(t_{k})\cdots V(t_{1})\ . \label{app2-eq:dyson-series-def} 
\end{align}
\end{subequations}
The convergence of the two series on their respective Hilbert spaces is the content of the following Lemma.
\begin{lemma}
\label{app2-prop:dyson}
The Dyson serie $\mathcal{D}(t)$ in~\eqref{app2-eq:dyson-series-def} converges in norm  to $U(t)U_0\dag(t)$ and $\mathcal{D}^N(t)$ in~\eqref{DysonN}
to $U^N(t)\bigl(U^N_0(t)\bigr)\dag$, the latter convergence being uniform with respect to $N$.
\end{lemma}

\begin{proof}
The proofs of the convergence of the two series are identical, thus we provide only that of $\mathcal{D}^N(t)$. Given the partial sums 
$$
D^N_{K}(t) \equiv 1 + \sum_{k=1}^{K}(-i)^{k}	\int_0^t dt_1\cdots \int_0^{t_{K-1}} dt_{K} V^N(t_{k})\cdots V^N(t_{1})\	,
$$
one has to show that that $\lim_{K\to\infty}\norm{U^N(t)\bigl(U^N_0(t)\bigr)\dag-D^N_K(t)} = 0$.
Writing
\begin{align*}
U^N(t)-U^N_0(t) &= \int_0^t dt_1 \frac{d}{dt_1}\bigl(U^N(t_1)U^N_0(t-t_1)\bigr)\\
&=-i \int_0^t dt_1 U^N(t_1)H^N_1U^N_0(t-t_1)\ ,
\end{align*}
and inserting~\eqref{app1-eq:V^N(t)-def}, one finds:
$$
U^N(t)\bigl(U^N_0(t)\bigr)\dag	= 1-i\int_0^t dt_1 U^N(t_1)\bigl(U^N_0(t_1)\bigr)\dag V^N(t_1) = D^N_{K}(t) + \Theta^N_{K+1}(t)\ ,
$$
where the remainder takes the form
\begin{equation*}
\Theta^N_{K+1}(t) \equiv (-i)^{K+1}\int_0^t dt_1\cdots \int_0^{t_{K}} dt_{K+1} U^N(t_{K+1})\bigl(U^N_0(t_{K+1})\bigr)\dag V^N(t_{K+1})\cdots V^N(t_{1})\ .
\end{equation*}
Therefore,
\begin{equation}
\label{ineqaux}
\Bignorm{	U^N(t)U_0(t)\dag	-	D^N_{K}(t)	} = \Bignorm{\Theta^N_{K+1}(t)} \le 	\Bignorm{H^N_1}^{K+1}\frac{t^{K+1}}{(K+1)!}\ .
\end{equation}
From~\eqref{H-pieces} and~\eqref{app2-eq:norm-E^N}, it follows that $H^N_1=H^N_\mathrm{int}$ is uniformly bounded with respect to $N$:
\begin{equation}
\label{app2-eq:norm-H^N_1-bound}
\Bignorm{H^N_1} = \lambda \mathfrak{c}_L\mathfrak{c}_R\Bignorm{\mathfrak{E}_+^N+\mathfrak{E}_-^N} \le 2 \lambda\ .
\end{equation}
Thus, when $K\to\infty$, the right hand side of the inequality~\eqref{ineqaux} vanishes uniformly with respect to $N$.
\end{proof}

\subsection{Conclusion of the proof of Theorem~\ref{ch5.2-th:dyn}}

We can now conclude the proof of Theorem~\ref{ch5.2-th:dyn}, which states that the evolution operator $U^N(t)$ converges in the mesoscopic limit 
to $U(t)$ (see~\eqref{mesoscopic-dynamics}). In order to do so, we first show that each term of the finite-size Dyson series $\mathcal{D}^N(t)$ converges to the corresponding one of $\mathcal{D}(t)$.

\begin{proposition}
\label{app1-prop:prod-vt-explicit}
For any choice of $k\in\N$ and of times $t_1,\ldots, t_k$, one has that
\begin{subequations}
\begin{align}
    & V^N(t_k)\cdots V^N(t_1) = \Bigl(\frac{\mathcal{E}_\mathrm{J}}{2}\Bigr)^k	\biggl[\sum_{\gamma_1=\pm}	\!\! \cdots \!\!\sum_{\gamma_k=\pm} \prod_{j=1}^k 
   (\mathfrak{E}^N_{\gamma_{k-j+1}}(t_{k-j+1})) \prod_{j=1}^ke^{ -it_j \Delta\xi(p^{N}_\varphi+\overline{\gamma}_{j-1},\gamma_j) } \biggr] \ , \label{eq:product-VNt-full-expression}\\
 & V(t_k)\cdots V(t_1) = \Bigl(\frac{\mathcal{E}_\mathrm{J}}{2}\Bigr)^k	\biggl[\sum_{\gamma_1=\pm}	\!\! \cdots \!\!\sum_{\gamma_k=\pm} \Bigl(\prod_{j=1}^k e^{i\gamma_{k-j+1}\varphi}\Bigr) \prod_{j=1}^ke^{ -it_j \Delta\xi(p_\varphi+\overline{\gamma}_{j-1},\gamma_j) } \biggr]\ , \label{eq:product-Vt-full-expression}
\end{align}
\end{subequations}
where we set $\overline{\gamma}_j \defsym \sum_{i=1}^j \gamma_j$. Furthermore, the following mesoscopic limit holds:
\label{final-prop}
\begin{equation}
\label{eq:mlim-dyson-term-by-term}
        \mlim_{N\to\infty} V^N(t_{k})\cdots V^N(t_{1})
        = V(t_{k})\cdots V(t_{1}) .
\end{equation}
\end{proposition}
\begin{proof}
We prove~\eqref{eq:product-Vt-full-expression}, the $N$-dependent case in~\eqref{eq:product-VNt-full-expression} being analogous. The derivation relies only on \eqref{app2-eq:comm-rels-a}, \eqref{app2-eq:comm-rels-b} and \eqref{app2-eq:comm-rels-time-dep}
using induction.
To start, set $k=1$. Using the Baker-Campbell-Hausdorff formula and applying~\eqref{app2-eq:comm-rels-a}, 
together with the definition 
of $\Delta\xi$ in in~\eqref{app2-eq:dxi-def}, one obtains
$$
U_0(t)\, e^{\pm i\hat\varphi}\, U_0\dag(t) = e^{\pm i\hat\varphi}e^{-it\Delta\xi(\hat p_\varphi,\pm1)}\ ,
$$
where also the identity:
\begin{equation*}
\Delta\xi(\hat p_{\varphi}-1,1) = \xi(\hat p_{\varphi})-\xi(\hat p_{\varphi}-1) = -\Delta\xi(\hat p_{\varphi},-1)\ ,
\end{equation*}
has been used.
Substituting the expression into the definition of $V(t)$ yields
\begin{equation*}\label{app1-eq:asdkcgvads}
    V(t) = \frac{\mathcal{E}_\mathrm{J}}{2}\sum_{\gamma=\pm} e^{i\gamma\hat\varphi} e^{-it\Delta\xi(\hat p_{\varphi},\gamma)},
\end{equation*}
which indeed coincides with what one has to prove for $k=1$.
On the other hand,
\begin{equation}\label{sec:3-th:1-eq:2}
V(t_{k})V(t_{k-1})\dots V(t_1)=
\frac{\mathcal{E}_\mathrm{J}}{2}\sum_{\gamma_{k}=\pm}e^{i\gamma_{k}\hat\varphi}e^{-it_{k}\Delta\xi(\hat p_\varphi,\gamma_{k})}	V(t_{k-1})\dots V(t_1),
\end{equation}
where we substituted for $V(t_{k})$ the previously obtained expression. The result follows by assuming $V(t_{k-1})\dots V(t_1)$ 
to have the desired form and then using~\eqref{app2-eq:dxi-eiphi-commutation} to bring the exponential generated by $\Delta\xi$ to the right of the formula.

We want now to show the validity of the mesoscopic limit in~\eqref{eq:mlim-dyson-term-by-term}. Using the expression~\eqref{eq:product-Vt-full-expression} which we have just derived, we have for $n_L,\,n_R,\,n'_L,\,n'_R\in\Z$,
\begin{multline}\label{app2-prop:final-eq:1}
	\braketmeso{V(t_{k})\cdots V(t_{1})} \\ 
	= \Bigl(\frac{\mathcal{E}_\mathrm{J}}{2}\Bigr)^k 
	\biggl(	\sum_{\gamma_1=\pm}\!\!	\dots	\!\!	\sum_{\gamma_k=\pm}	\delta_{n',n+\overline{\gamma}_k}
	\prod_{j=1}^k e^{-it_j\Delta\xi(n+\overline{\gamma}_{j-1},\gamma_j)} \Bigr)\ \delta_{n_L'+n_R',n_L+n_R} \ ,
\end{multline}
where we also set, in analogy with~\eqref{app2-eq:n-def},
\begin{equation}
    n'= \frac{n_L'-n_R'}{2}\ .
\end{equation}
On the other hand, the expression~\ref{eq:product-VNt-full-expression} yields
\begin{multline}\label{app2-prop:final-eq:3}
    \braketgns{{V}^N(t_{k})\cdots {V}^N(t_{1})}   \\
    = \Bigl(\frac{\mathcal{E}_\mathrm{J}}{2}\Bigr)^k	\biggl(\sum_{\gamma_1=\pm}	\!\! \cdots \!\!\sum_{\gamma_k=\pm} \Braketgns{\prod_{j=1}^k 
  \mathfrak{E}_{\gamma_{k-j+1}}^N(t_{k-j+1})} \prod_{j=1}^ke^{ -it_j \Delta\xi(n+\overline{\gamma}_{j-1},\gamma_j) } \biggr)\ .
\end{multline}
The limit follows then directly from~\eqref{app2-eq:time-ev-phase-op-vanish}, which states that the product of an arbitrary number of phase operators $\mathfrak{E}_{\pm}^N(t)$ converges to the corresponding product of Weyl operators in the Heisenberg algebra on the circle.
\end{proof}

\noindent

\dynamicsqubit*
\begin{proof}
Let us define the quantity
\begin{gather*}
I^N(t) \equiv \Bigabs{\braketgns{U^N(t)} - \braketmeso{U(t)}},
\end{gather*}
where once again $n_L,\,n_R,\,n'_L,\,n'_R\in\Z$. We want to show that $I^N(t)$ goes to zero as $N$ grows large. Thanks to the results of Lemma~\ref{app2-prop:dyson}, we can replace the evolution operators $U^N(t)$ and $U(t)$ with their respective Dyson series, which converge in norm to the evolution operators. Therefore, for a chosen $\epsilon>0$,  we can truncate the Dyson series at such a large $K$, that independently of $N$,
$$
\Bignorm{U^N(t)-\mathcal{D}^N_K(t)U^N_0(t)}\leq\frac{\epsilon}{3}\ ,\quad \Bignorm{U(t)-\mathcal{D}_K(t)U_0(t)}\leq\frac{\epsilon}{3}\ ,
$$
and estimate $\displaystyle I^N(t)\leq I^N_K(t)\,+\,2\,\frac{\epsilon}{3}$, where
\begin{gather*}
 I^N_K(t) \equiv \Bigabs{\braketgns{\mathcal{D}^N_K(t)U^N_0(t)} - \braketmeso{\mathcal{D}_K(t)U_0(t)}}.
\end{gather*}
Notice that
\begin{equation}
\label{eq:INK_bound}
\begin{split}
I^N_K(t) \, \le \, \sum_{k=1}^K \int_0^t dt_1\cdots \int_0^{t_{k-1}} dt_{k}\, \Bigabs{J^N(t_1,\cdots,t_k)}
\end{split}
\end{equation}
where, since $U^N_0(t) = U^N_\mathrm{free}(t)\, U^N_\mathrm{C}(t)$ (see~\eqref{app2-eq:micro-free-evolutor-factorization}),
\begin{equation*}
\begin{split}
\Bigabs{J^N(t_1,\cdots,t_k)} \equiv \, &\Bigl\lvert \braketgns{V^N(t_{k})\cdots V^N(t_{1})U^N_\mathrm{free}(t)\, U^N_\mathrm{C}(t)} - \\
					& - \braketmeso{V(t_{k})\cdots V(t_{1})U_0(t)} \Bigr\rvert.
\end{split}
\end{equation*}
As we shall shortly see, by choosing $N$ large enough, one can make $\Bigabs{J^N(t_1,\cdots,t_k)}$ arbitrarily small so that, because of the finite sum and finite time $t$ in the right hand side of~\eqref{eq:INK_bound},  $\displaystyle I^N_K(t)\leq\frac{\epsilon}{3}$ and thus $I^N(t)\leq\epsilon$.

Indeed, acting on the corresponding vectors, $U^N_\mathrm{C}(t)$ and $U_0(t)$ give rise to the same phase factor, which thus drops out. We can then proceed with adding and subtracting the same term 
\begin{equation}
\label{eq:J-bound}
\begin{split}
\Bigabs{J^N(t_1,\cdots,t_k)}	=	& \, \Bigabs{ J^N(t_1,\cdots,t_k) + \braketgns{V^N(t_{k})\cdots V^N(t_{1})} - \\
								& - \braketgns{V^N(t_{k})\cdots V^N(t_{1})} } \\
							\le	& \, \Bigabs{ J^N_1(t_1,\cdots,t_k) } + \Bigabs{ J^N_2(t_1,\cdots,t_k) },
\end{split}
\end{equation}
where, using~\eqref{app2-eq:bcs-vanish} and Proposition~\ref{app1-prop:prod-vt-explicit}, the following two quantities
\begin{gather*}
J^N_1(t_1,\cdots,t_k) \equiv \braketgns{V^N(t_{k})\cdots V^N(t_{1})\bigl(U^N_\mathrm{free}(t)-1\bigr)}; \\
J^N_2(t_1,\cdots,t_k) \equiv \braketgns{V^N(t_{k})\cdots V^N(t_{1})} - \braketmeso{V(t_{k})\cdots V(t_{1})}.
\end{gather*}
can be made arbitrarily small by choosing $N$ large enough . 
\end{proof}